\documentclass[3p,authoryear]{elsarticle}
\usepackage[utf8]{inputenc}
\usepackage{caption}
\usepackage{subcaption}
\usepackage{graphicx}
\usepackage{soul, color}
\usepackage{amsfonts}
\usepackage{calrsfs}
\usepackage{hyperref}
\DeclareMathAlphabet{\pazocal}{OMS}{zplm}{m}{n}
\usepackage{amsmath}
\usepackage{bbm}
\DeclareMathOperator*{\argmax}{arg\,max}

\usepackage{bm}
\usepackage{amsthm}
\usepackage{upgreek}
\usepackage{empheq} 
\usepackage{booktabs}
\usepackage{multirow}
\usepackage{tabularx}
\usepackage{algorithm}
\usepackage{algpseudocode}
\usepackage{tcolorbox}
\usepackage{ulem}
{\color{yellow}}%
{}

\usepackage{hyperref}
\newtheorem{theorem}{Theorem}

\newtheorem{definition}{Definition}

\begin{document}
% Put here what will go to headers as author
\title{OCC-MP: A Max-Pressure framework to prioritize transit and high occupancy vehicles}

% TODO: add macros for easier formatting of \author.
\author[1]{Tanveer Ahmed}
\ead{tpa5285@psu.edu}
\author[2]{Hao Liu\corref{cor1}}
\ead{hfl5376@psu.edu}
\cortext[cor1]{Corresponding author}
\author[1]{Vikash V. Gayah}
\ead{gayah@engr.psu.edu}
\address[1]{Department of Civil and Environmental Engineering, The Pennsylvania State University, University Park, PA, United States}

\address[2]{Department of Civil and Environmental Engineering, Jackson State University, Jackson, MS, United States}

% If necessary modify the number of words per table or figure default is set to
% 250 words per table
% \WordsPerTable{250}

% If words are counted manually, put that number here. This does not include
% figures and tables. This can also be used to avoid problems with texcount
% program i.e. if one does not have it installed.
% \TotalWords{200}

\begin{abstract}
Max-pressure (MP) is a decentralized adaptive traffic signal control approach that has been shown to maximize throughput for private vehicles. However, MP-based signal control algorithms do not differentiate the movement of transit vehicles from private vehicles or between high and single-occupancy private vehicles. Prioritizing the movement of transit or other high occupancy vehicles (HOVs) is vital to reduce congestion and improve the reliability and efficiency of transit operations. This study proposes OCC-MP: a novel MP-based algorithm that considers both vehicle queues and passenger occupancies in computing the weights of movements. By weighing movements with higher passenger occupancies more heavily, transit and other HOVs are implicitly provided with priority, while accounting for any negative impacts of that priority on single occupancy vehicles. And, unlike rule-based transit signal priority (TSP) strategies, OCC-MP more naturally also accommodates conflicting transit routes at a signalized intersection and facilitates their movement, even in mixed traffic without dedicated lanes. Simulations on a grid network under varying demands and transit configurations demonstrate the effectiveness of OCC-MP at providing TSP while simultaneously reducing the negative impact imparted onto lower occupancy private vehicles. Furthermore, OCC-MP is shown to have a larger stable region for demand compared to rule-based TSP strategies integrated into the MP framework. The performance of OCC-MP is also shown to be robust to errors in passenger occupancy information from transit vehicles and can be applied when passenger occupancies of private vehicles are not available. Finally, OCC-MP can be applied in a partially connected vehicle (CV) environment when a subset of vehicles is able to provide information to the signal controller, outperforming baseline methods at low CV penetration rates.
\end{abstract}


\maketitle

\section{Introduction}
Adaptive Traffic Signal Control (ATSC) is an intelligent transportation system technology that aims to optimize traffic flow by dynamically adjusting signal timings based on real-time traffic patterns. Max Pressure (MP) is a decentralized ATSC framework that has gained popularity due to its effectiveness in improving vehicle throughput at intersections. Initially developed for packet transmission scheduling in wireless networks \citep{tassiulas1990a}, the MP concept was later extended to traffic signal control by \citep{varaiya2013a}. MP-based traffic signal control algorithms operate independently at each intersection and rely on local information from approach links upstream and downstream of the intersection. Unlike some other ATSC approaches, MP algorithms do not require knowledge of future traffic demands, making them more practical and applicable in real-world settings. MP control is based on distributing vehicles from longer queues to shorter queues \citep{levin2023a}. Specifically, the control policy assigns the right of way to the phase in a traffic signal that serves movements with higher level of congestion, toward downstream links that are more uncongested in order to maximize throughput. The level of congestion can be measures using various metrics such as the number of vehicles, average travel time or average delay along a link \citep{dixit2020a, kouvelas2014a, le2015a, lioris2016a, liu2022novel, liu2023total, mercader2020a, varaiya2013a, xiao2014a}. The most desirable property of the MP algorithm is maximum stability, which refers to its ability to serve a feasible set of demands if those demands can be accommodated by any other control strategy \citep{varaiya2013a}. The set of feasible demands is known more commonly as the stable region. Modifications to the MP in literature have reported either maximum stability or stability properties on a reduced stable region \citep{barman2023a, gregoire2014a, le2015a, levin2020a, li2019a, liu2022novel, pumir2015a, wu2018a, xiao2014a, xu2022a}. While there have been several variations of the MP algorithm proposed since 2013, most treat all vehicles in a similar manner and do not distinguish between low and high occupancy vehicles (HOVs). However, prioritizing HOVs -- particularly transit vehicles -- is critical to reduce vehicular demand and alleviate congestion.

To that end, transit signal priority (TSP) aims to enhance the performance of public transportation by granting priority to transit vehicles at traffic signals. Its primary objective is to alleviate delays caused by traffic signals and improve the reliability, efficiency, and speed of public transportation services. TSP techniques can generally be classified into three categories: “active”, “passive”, and "adaptive". Passive TSP relies on pre-programmed signal timing plans to prioritize public transit vehicles at specific times or on designated routes, without direct communication between the transit vehicle and traffic signals. It is effective for fixed-route bus lines with predictable schedules \citep{lin2019a, stephanedes1996a}. In contrast, active TSP involves real-time communication between transit vehicles and traffic signals, allowing for dynamic adjustments to signal timing based on the vehicle’s needs. Active TSP requires a two-way communication system, with transit vehicles sending requests to the traffic signal system, which then responds by adjusting signal timing through methods like green extension and red truncation \citep{christofa2011a, currie2008a, lin2015a, truong2019a, zeng2021a}. Most of these studies have focused on developing TSP strategies based on fixed cycle lengths or are limited to dedicated bus lanes. As a result, these strategies overlook the potential consequences on private vehicles i.e., overall traffic flow. In addition, these studies rely on rule-based approaches and optimization under various constraints to balance travel time of transit and private vehicles \citep{lee2022a}.

Adaptive-TSP systems dynamically respond to changing traffic conditions and adjust signal timings accordingly. These systems utilize real-time traffic information and other vehicle information -- including vehicle occupancies -- to optimize performance measures, such as minimizing delay for both vehicles and passengers. In many cases, the optimization objectives consider prioritizing factors like maximizing person capacity or minimizing person delay, schedule delay, vehicle queues and emissions rather than vehicle-based measures alone \citep{chen2023a, christofa2016a, christofa2013a, christofa2011b, ding2015a, lee2022a, li2020a, ma2014a, yu2018a, yu2017a, zeng2015a, zhao2018a, zhao2019a}. The computational complexity of these problems calls for formulation as mixed integer linear problems that are commonly solved using techniques such as dynamic programming \citep{wu2019a, wu2020a}, genetic algorithms \citep{ghanim2015a, yang2023a}, reinforcement learning \citep{alizadeh2019a, guo2021a, ling2004a, long2022a}. With the emergence of connected vehicle (CV) technology, researchers have leveraged two-way communication, precise vehicle location tracking, and passenger count information in TSP research \citep{chen2022a, hu2022a, hu2015a, yang2019a, yang2023a}. More recently, passenger occupancy-based signal timing has been developed using reinforcement learning; see e.g., \citep{vlachogiannis2023humanlight}. However, these learning-based methods rely on extensive training and are not generally transferable to situations that were not observed within the training process. And, as will be described later in this paper, the reward considered in \citep{vlachogiannis2023humanlight} does not appropriately account for downstream space availability, which could limit its effectiveness.  

A recent study \citep{xu2022a} proposed the integration of rule-based TSP into MP control. The study demonstrated that the combined MP-TSP policy outperforms fixed-time-TSP and adaptive-TSP policies in reducing vehicle travel times, while having the ability to serve a larger demand. However, the proposed policy relies on constraints that reduce the stable region or private vehicles compared to the original max pressure policy. Moreover, the control uses a set of rules that switch between the original control policy and TSP depending on the detection of buses thus, it provided unconditional priority to buses at the expense of other vehicles. Furthermore, the application was also limited to arterials with dedicated bus lanes, while most urban networks have shared lanes for transit vehicles and private automobiles. 

In light of these drawbacks, this study proposes a novel occupancy-based MP policy (OCC-MP) that combines passenger occupancies and vehicle queues when determining signal timings. By increasing the weight of movements with more queued passengers in the signal timing process, transit and other HOVs are implicitly prioritized over private automobiles as they typically carry more passengers. The study also analytically proves that OCC-MP maintains the maximum stability property for isolated intersections without reducing the size of the stable region. A stability analysis using micro-simulation also reveals that the OCC-MP policy has a similar stable region to the original MP, as well as a larger stable region compared to rule-based MP that provides TSP. Simulation tests demonstrate that OCC-MP can provide priority to higher occupancy private automobiles, if occupancy information of these vehicles is available. Unlike previous attempts to integrate TSP with MP, the proposed strategy can also be applied to networks with shared bus lanes (i.e., transit vehicles and private automobiles move in the same lanes).  OCC-MP can also be applied in a partial CV environment, both when a subset of vehicles can be detected and when a subset can provide occupancy information to the signal controller. The performance of OCC-MP is also shown to be robust to errors in passenger occupancy information from transit vehicles. 

The remainder of this paper is organized as follows. The next section provides an overview of the proposed OCC-MP control policy. This is followed by the simulation setup used to evaluate the performance of OCC-MP against baseline methods. Then, the results of the experiments are presented, including a comparative analysis between the proposed methods and the baseline approaches. The last section concludes the paper by highlighting the important findings and suggesting potential directions for future research.

\section{Method}
\subsection{Max Pressure}
Before the MP signal control is described, some notations are provided. Consider a network made up of links and nodes: each link represents a directional road segment between two adjacent intersections which are represented by nodes. Upstream and downstream links at an intersection facilitate the movement of incoming and outgoing vehicles. Movement $\left(l,m\right)$ represents the pair of links that serves vehicles from an upstream link $l$ to a downstream link $m$. $U(i)$ denotes the set of all upstream links at intersection $i$, and $D(l)$ denotes the set of links downstream of link $l$. The proportion of traffic that turns from link $l$ to link m is denoted by $r(l,m)$. The rate vehicles are allowed to pass through an intersection from link $l$ to m per unit time is represented by the saturation flow, $C(l,m)$, which is a random but upper bounded variable. Each intersection serves a set of signal phases denoted by $\Phi_i$  where each signal phase serves a set of vehicular movements. $L_i^\phi$ contains the set of movements served by phase $\phi$ at intersection $i$. 

The MP algorithm involves three key steps; 
\begin{enumerate}
 \item Obtain the weights $(w)$ of each movement. Weight is assigned to each movement by calculating the difference between the metric value of that movement and the average value of the metric for its downstream movements. This weight serves as an indicator of the level of congestion of both the upstream and downstream end of a movement. Note that for isolated intersections or intersections at the boundary of a network, downstream movements serving exit vehicles are not considered.
	
 \item Calculate the pressure $(P)$ of each phase. The pressure of each phase is calculated by summing up the weight multiplied by the associated saturation flow over all movements served by that phase. This is used to determine the relative importance of each phase served by the signal. 
	
 \item Determine signal timing $(S)$ using pressure. In acyclic MP algorithms, the phase with the highest pressure is activated in the next time step without regarding the sequence of phases. On the other hand, in cyclic MP algorithms, the green time for each phase in the next cycle is assigned proportionally based on the pressures of the respective phases in the designated phase sequence. The proposed model follows the former type.
\end{enumerate}

\subsection{Proposed OCC-MP policy}
The original MP policy proposed in \citep{varaiya2013a} -- referred to in this paper as the Q-MP -- uses the number of queued vehicles on each link as the metric to determine the weights of the movements. Thus, it treats both buses and private vehicles equally and disregards the fact that a bus can transport significantly more passengers compared to a single-occupancy passenger vehicle. Consequently, in the Q-MP algorithm, the right of way may be assigned to a movement with five single-occupancy vehicles rather than a bus carrying fifty passengers. In contrast, traditional rule-based TSP algorithms (including that integrated with MP in \citep{xu2022a}) prioritize bus movements at an intersection regardless of the level of congestion on adjacent links. This means that a bus with no passengers would be given the right of way over a conflicting movement with many queued vehicles posing the risk of a queue spillback.

To address these limitations, the proposed algorithm (OCC-MP) considers the number of queued vehicles and the average occupancy of the vehicles queued on the upstream movements in order to prioritize movements involving transit or HOVs. Specifically, the weight assigned to each movement is calculated as the product of the difference between the upstream and downstream queue lengths and the average occupancy upstream (\ref{eq:weight_occmp}).

\begin{equation}\label{eq:weight_occmp}
w(l,m)=o(l,m)\left[x(l,m)-\sum_{n\in D(m)}x(m,n)r(m,n)\right]^{+} =o(l,m)w_q (l,m)^{+}
\end{equation}

where $o(l,m)$ is the average occupancy over all vehicles in the upstream movement $(l,m)$; $x(i,j)$ is the number of vehicles queued on movement $(i,j)$, and the $+$ symbol around the square bracket denotes the maximum of either 0 or the value inside the square brackets. The term in the square brackets is the weight of the original MP (Q-MP),  $w_q (l,m)$. An additional modification is made so that movements with negative weights are set to 0. This arises when downstream links are more congested than upstream links. Ignoring negative weights is useful when a phase serves multiple movements -- e.g., a through and a right turn movement -- and the weight of one movement is negative while the other is positive. Often, a negative weight on a minor movement (e.g., a right turn) adversely reduces the weights of the major movement (e.g., the through movement). This reduces the pressure of that phase and leads to activation of less critical phases, reducing throughput. Such a modification has been used in prior studies and shown to improve network performance \citep{gregoire2015a, AHMED2024122}. This term is then multiplied by the average occupancy of vehicles upstream. Since all vehicles have an occupancy of at least 1 and the maximum passenger capacity is finite, the average occupancy on a movement is a positive and bounded number. Therefore, the weight calculation in the OCC-MP algorithm is essentially a scaled-up version of the weight defined in the Q-MP algorithm.

At each update interval, the pressure of phase $\phi$ can be expressed as:
\begin{equation}\label{eq:pressure_occmp}
P^{\phi}=\sum_{(l,m)\in L^{\phi}_i}w(l,m)\times C(l,m)\times S(l,m), \quad \forall \phi\in \Phi_i
\end{equation}
where $S(l,m)$ is a binary variable associated with phase $\phi$ indicating whether movement $(l,m)$ is served by phase $\phi$.

Finally, at intersection $i$, the policy selects the phase with the maximum pressure in the set of all phases $\Phi_i$ (\ref{eq:signal_occmp}). In this study, the signals are updated in the subsequent time step.  

\begin{equation}\label{eq:signal_occmp}
S^*=\argmax_{\phi\in \Phi_i}P^{\phi}
\end{equation}

The benefit of considering the average occupancy is that it allows the control policy to distinguish between movements that serve vehicles of higher occupancy and those that do not. Figure \ref{fig:figure_1} provides a simple illustrative example with two one-way movements. Only private vehicles are queued in the N-S direction while both private vehicles and a bus are queued in the W-E direction. The W-E movement has 3 vehicles queued upstream and 2 vehicles downstream; thus its weight under the Q-MP policy is $w_q (W,E)=(3-2)=1$. The N-S movement has 5 queued vehicles on its upstream link and 2 vehicles downstream, hence, its weight under the Q-MP policy is $w_q (N,S)=(5-2)=3$. Therefore, Q-MP prioritizes the N-S movement over the W-E movement. Under the OCC-MP policy however, the W-E movement a weight of, $w(W,E)=(20+2+2)/3×(3-2)=8$ while the N-S movement has a weight of $w(N,S)=w_q (N,S)=(1+1+1+1+1)/5×(5-2)=3$. Therefore, OCC-MP provides TSP by prioritizing the W-E movement.

\begin{figure}[!ht]
    \centering
        \includegraphics[width=3in]{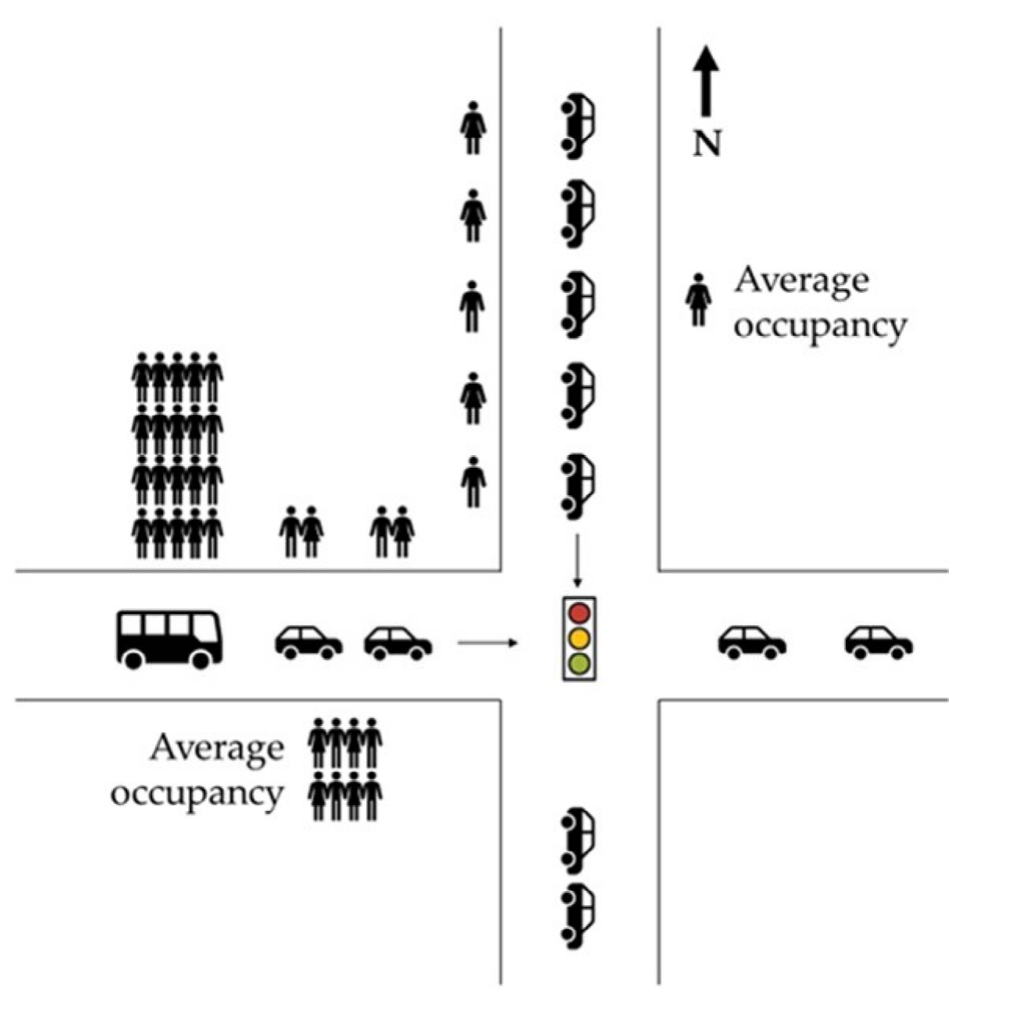}
    	\centering
    	\caption{Example of transit signal priority using OCC-MP.}
    	\label{fig:figure_1}
\end{figure}

The occupancy of the downstream vehicles is not considered when calculating the weight of the movement. This is because the downstream portion of the weight captures supply constraints on that link. Therefore, simply replacing the number of queued vehicles with the number of passengers on both upstream and downstream -- as was done in \citep{vlachogiannis2023humanlight} -- would mean downstream supply issues. The presence of downstream vehicles accounts for available storage space on the receiving links. More waiting passengers downstream, specifically in buses, does not necessarily mean that a link has little capacity to accommodate vehicles from upstream links. Therefore, only the average upstream occupancy is considered. In cases where there are no vehicles downstream or in isolated intersections where downstream movements are ignored in the weight calculation, the weight of the movements in the OCC-MP algorithm is equal to the number of passengers upstream. For example, in Figure \ref{fig:figure_1}, if there were no vehicles downstream on either movements, $w(N,S)$ and $w(W,E)$ would represent the number of queued passengers on the N-S and the W-E movements respectively. 

Since OCC-MP requires information on vehicle occupancy, it is assumed that such information is available to the controller. In scenarios where private vehicle occupancies are not readily available, an average occupancy value is assumed. However, in a fully connected vehicle environment, it is assumed that the occupancy information may be readily accessible. On the contrary, many buses are currently equipped with automatic passenger counting (APC) systems that allow real-time information of the number of passengers onboard a transit vehicle. Therefore, the exact bus occupancies are assumed to be available for weight calculation.

Intersections where conflicting bus routes are served by different phases often receive simultaneous priority requests. Most prior studies have used either a first-come-first-serve or model-based methods (e.g., person-delay optimization, schedule-deviation minimization) to decide the sequence of phases at conflict intersections \citep{christofa2013a, head2006a, hu2016a, ma2013a}. These methods are subject to strict constraints that reduce the efficacy and increase the complexity of TSP control policies. The proposed OCC-MP handles conflict intersections much more efficiently without any additional constraints or assumptions. Specifically, if multiple buses are competing for right of way, OCC-MP selects the phase with the highest pressure considering the size of the queue on the link and the average occupancy of both buses and private vehicles. This way, OCC-MP is able to resolve conflicting bus movements at intersections without compromising the flow of private vehicles.

\subsection{Maximum stability for isolated intersections}\label{subsec:max_stability}
A signal control policy is stable if the average service rate in the network is equal to the average demand, i.e., the average number of vehicles in the network remain bounded \citep{levin2023a}. Maximum stability refers to the property that the policy can serve a traffic demand if this demand can be accommodated by an admissible control strategy. This section proves the maximum stability property for the proposed OCC-MP for isolated intersections.

For isolated intersections, as mentioned before, all outgoing links are regarded as sink links, and the traffic states on those links are not considered in the pressure calculation. Therefore, only the incoming links for isolated intersections are considered for the rest of this section. Traffic dynamics, in terms of the evolution of number of vehicles, on movement $\left(l,m\right)$ can be expressed as:

\begin{equation}\label{eq:evolution}
x\left(l,m\right)\left(t+1\right)=x\left(l,m\right)\left(t\right)+d\left(l,m\right)\left(t\right)-\min{\{}C\left(l,m\right)\left(t\right)S\left(l,m\right)\left(t\right),x\left(l,m\right)\left(t\right)\}
\end{equation}
where $d\left(l,m\right)\left(t\right)$ is the external demand for movement $(l,m)$ at time $t$.

\begin{definition}\label{def:feasibledemand}
A demand $\mathbf{d}$ is feasible if there exists an admissible control sequence $\{\mathbf{S}(t):t=0,1,2,...\}$ such that
\begin{equation}\label{eq:feasible}
    \bar{d}(l,m)\le \bar{S}(l,m)c(l,m), \quad \forall (l,m)
\end{equation}
where $\bar{d}(l,m)$ is the average external demand of movement $(l,m)$, $\bar{S}(l,m)$ is the proportion of time steps that movement $(l,m)$ is activated, and $c(l,m)$ is the average saturation flow for movement $(l,m)$.  
\end{definition}

The set of demand satisfying Eq. \eqref{eq:feasible}, denoted by $\pazocal{D}$, is called feasible demand region, and $\pazocal{D}^0$ is used to indicate the interior of $\pazocal{D}$. Let $\pazocal{S}$ indicate the set of admissible phases at the isolated intersection and $co(\pazocal{S})$ denote the convex hull of $\pazocal{S}$, which can be expressed as:

\begin{equation}\label{eq:ch_phase}
    co(\pazocal{S})=\{\sum_{\mathbf{S}^{e}\in \pazocal{S}}\lambda_e\mathbf{S}^{e}|\lambda_e\ge 0, \sum_{e}\lambda_e=1\}
\end{equation}

It is easy to prove that a control matrix $\mathbf{\Sigma}$ is in $co(\pazocal{S})$ if and only if there exists an admissible control sequence $\{\mathbf{S}(t):t=0,1,2,...\}$ such that:

\begin{equation}\label{eq:hull}
    \bar{S}(l,m)=\mathbf{\Sigma}(l,m), \quad \forall (l,m)
\end{equation}

Combining Definition \ref{def:feasibledemand} and Eq. \eqref{eq:hull} obtains that a demand $\mathbf{d}$ is in the feasible region if and only if there exists a matrix $\mathbf{\Sigma}\in co(\pazocal{S})$ such that
\begin{equation}\label{eq:demand2}
    \bar{d}(l,m)\le \mathbf{\Sigma}(l,m)c(l,m), \quad \forall (l,m)
\end{equation}

\begin{definition}
    A signal control sequence $\{\mathbf{S}(t):t=0,1,2,..., T\}$ stabilizes the queue process in the mean if the average queue length in the network is upper bounded, i.e.,
    
    \begin{equation}\label{eq:def_ms_v}
        \frac{1}{T}\sum_{t=1}^T\sum_{l,m}E(x(l,m)(t))\le M, \quad T=1,2,3...
    \end{equation}
    where $M<\infty$.
\end{definition}

\begin{theorem}[Maximum stability]\label{theorem:ms}
    Assume the expected value of the average occupancy of a movement is independent of the queue length of that movement, the proposed OCC-MP algorithm stabilizes the queue process for isolated intersections if $\mathbf{d}\in \pazocal{D}^0$.
\end{theorem}

\begin{proof}
As mentioned before, it is assumed the supply downstream of the outgoing links at isolated intersections is infinity, so the number of vehicles on the outgoing links is always upper bounded. Therefore, only number of vehicles on the incoming links is considered in this proof.

Let $\mathbf{\delta}$ represent the difference in the queue length of movement $(l,m)$ between two consecutive steps under the control of OCC-MP, $\mathbf{S}^*$. Eq. \eqref{eq:evolution} leads to, 
\begin{equation}\label{eq:delta}
    \delta(l,m)(t+1)=x(l,m)(t+1)-x(l,m)(t)=d(l,m)(t)-\min(x(l,m)(t), C(l,m)(t)S^*(l,m)(t))
\end{equation}

Let $\sqrt{\mathbf{o}}(t)$ indicate the vector of which each element is the square root of the average occupancy of a movement at the isolated intersection, i.e., $\sqrt{\mathbf{o}}(t)=\{\sqrt{o(l,m)}: \forall (l,m)\}$. Then,
\begin{equation}\label{eq:keyequation}
    |\sqrt{\mathbf{o}}(t)\odot\mathbf{X}(t+1)|^2-|\sqrt{\mathbf{o}}(t)\odot\mathbf{X}(t)|^2=2(\sqrt{\mathbf{o}}(t)\odot\mathbf{X}(t))^T(\sqrt{\mathbf{o}}(t)\odot\mathbf{\delta}(t))+|\sqrt{\mathbf{o}}(t)\odot\delta(t+1)|^2=2\alpha+\beta
\end{equation}
where $\mathbf{A}\odot \mathbf{B}$ is the Hadamard product of vectors $\mathbf{A}$ and $\mathbf{B}$, $|\mathbf{X}|\equiv\sum_{x_i \in \mathbf{X}}x_i$,  $|\mathbf{X}|^2\equiv\sum_{x_i \in \mathbf{X}}(x_i)^2$. 

Next, it is proved that there exist $k<\infty$ and $\epsilon>0$ such that
\begin{equation}
    E\{(|\sqrt{\mathbf{o}}(t)\odot\mathbf{X}(t+1)|^2-|\sqrt{\mathbf{o}}(t)\odot\mathbf{X}(t)|^2)|\mathbf{X}(t), \sqrt{\mathbf{o}}(t)\}\le k-\epsilon|\mathbf{X}(t)|
\end{equation}

According to Eq. \eqref{eq:delta}, $\beta$ can be expressed as: 
\begin{equation}\label{eq:beta}
    \beta=\sum_{l,m}o(l,m)(t)\left[d(l,m)(t)-\min(x(l,m)(t), C(l,m)(t)S^*(l,m)(t))\right]^2
\end{equation}

Since $o(l,m)(t)$, $d(l,m)(t)$, and $\min(x(l,m)(t), C(l,m)(t)S^*(l,m)(t))$ are all upper bounded by a constant, it is easy to see that $\beta$ is upper bounded by a constant.

For simplicity, $(l,m)$ from the expression of $\alpha$ is omitted in the following. Combining Eqs. \eqref{eq:delta} and \eqref{eq:keyequation} obtains,
\begin{equation}\label{eq:beta}
\begin{split}
    \alpha& =\sum_{l,m}o(t)x(t)(d(t)-\min(x(t), C(t)S^*(t))\\
    & = \sum_{l,m}o(t)x(t)(d(t)-C(t)S^*(t))+\sum_{l,m}o(t)x(t)(C(t)S^*(t)-\min(x(t), C(t)S^*(t)))\\
    & = \sum_{l,m}\alpha^{'}_1(l,m)(t)+\sum_{l,m}\alpha^{'}_2(l,m)(t)\\
    & = \alpha_1+\alpha_2
\end{split}
\end{equation}

It is easy to obtain that for each movement $(l,m)$,

\begin{subequations}\label{eq:alpha2}
	\begin{empheq}[left={\alpha^{'}_2(l,m)(t)\empheqlbrace\,}]{alignat=2}
	& =0, && \quad \text{if $x(l,m)(t) \ge C(l,m)(t)S^*(l,m)(t)$}\\
	& \le o(l,m)(t)C^2(l,m)(t) && \quad \text{otherwise}
	\end{empheq}
\end{subequations}

Therefore, $\alpha_2=\sum_{l,m}\alpha^{'}_2(l,m)(t)$ is also upper bounded by a constant.  

Since $S^*$ maximizes $\sum_{l,m}o(t)x(t)C(t)S(t)$ and $d(t)$ is in the stable region, according to Eq. \eqref{eq:demand2}, there exist a control matrix $\mathbf{\Sigma}\in co(\pazocal{S})$ and $\epsilon>0$ such that
\begin{equation}
    \begin{split}
        E\{\alpha_1|\mathbf{X}(t), \sqrt{\mathbf{o}}(t)\}&\le \sum_{l,m}o(t)x(t)E\{(d(t)-C(t)\Sigma)|\mathbf{X}(t), \sqrt{\mathbf{o}}(t)\}\\
        & = \sum_{l,m}o(t)x(t)(\bar{d}-c(t)\Sigma)\\
        & \le -\sum_{l,m}\epsilon o(t)x(t)
    \end{split}
\end{equation}

Since $o(l,m)>=1$, 
\begin{equation}\label{eq:alpha1bound}
E\{\alpha_1|\mathbf{X}(t), \sqrt{\mathbf{o}}(t)\} \le -\sum_{l,m}\epsilon x(t)
\end{equation}

Until now, it is proved that $\beta$ and $\alpha_2$ are upper bounded by a constant, and the conditional expectation of $\alpha_1$ is bounded by a form shown in Eq. \eqref{eq:alpha1bound}. Then, 
\begin{equation}
\begin{split}
    & E\{(|\sqrt{\mathbf{o}}(t)\odot\mathbf{X}(t+1)|^2-|\sqrt{\mathbf{o}}(t)\odot\mathbf{X}(t)|^2)|\mathbf{X}(t), \sqrt{\mathbf{o}}(t)\}\\
    = & E\{\mathbf{o}(t)\}^TE\{(\mathbf{X}(t+1)\odot\mathbf{X}(t+1)-\mathbf{X}(t)\odot\mathbf{X}(t)|\mathbf{X}(t), \sqrt{\mathbf{o}}(t)\}\\
    \le & k-\epsilon|\mathbf{X}(t)|
\end{split}
\end{equation}

Then, taking expectation for both side and sum over $t=1, 2, 3, ..., T$ obtains:
\begin{equation}\label{eq:laststep}
    \begin{split}
        \frac{\sum_{t=1}^{T}E\{|\mathbf{X}(t)|\}}{T}\le \frac{k}{\epsilon}+\frac{1}{\epsilon T}E\{\mathbf{o}\}^TE\{\mathbf{X}(1)\odot\mathbf{X}(1)\}
    \end{split}
\end{equation}

Since it is assumed that the expectation of average occupancy, $E\{\mathbf{o}\}$, is fixed, the right-hand side of Eq. \eqref{eq:laststep} is a constant. Eq. \eqref{eq:laststep} indicates the average queue length in the network under the control of OCC-MP is upper bounded.
\end{proof}

Although the maximum stability for a large network is not established for the original feasible region, defined by Q-MP, the simulation results in Section \ref{sec:results} demonstrate that the feasible region of the proposed OCC-MP is larger than that of a rule-based MP algorithm that provides TSP \citep{xu2022a}, which was proved to be stable for a smaller feasible region. 

\section{Simulation Setup}\label{sec:simset}
Simulation tests in the AIMSUN micro-simulation platform were performed to evaluate the effectiveness of the proposed OCC-MP control strategy. AIMSUN was chosen for its ability to realistically model traffic dynamics, such as congestion propagation, queue spillbacks, vehicle routing, and driving behavior\citep{barcel2005a}.
\subsection{Network setup}
Simulation tests were carried out on an 8x8 grid network shown in Figure \ref{fig:figure_2}. While real-world street networks may not perfectly align with a square grid pattern, many urban networks exhibit grid-like characteristics. Previous studies that have simulated grid traffic networks reported results that can be generalized to more realistic networks \citep{bayrak2023a, knoop2012a, mazloumian2010a, ortigosa2019a, ortigosa2014a}. Road segments were categorized as arterials with mixed use that accommodate both private vehicles and buses. All road segments were assumed to have bi-directional traffic flow, with three travel lanes in each direction serving dedicated right, through and left movements at an approach (Figure \ref{fig:figure_3}). Each segment was 200 meters long with a saturation flow of 1800 vehicles per hour and a posted speed limit of 50 km/h. Within the network, all 64 intersections were signalized and consist of four phases, where through and right movements are served by one phase while left turning movements have a separate phase (Figure \ref{fig:figure_3}). 

Private vehicle origins and destinations were strategically positioned at the 32 O-D centroids located along the network’s perimeter. A symmetric demand pattern was adopted, in which the demand at North-South origin centroids was assumed to be twice the demand at East-West origin centroids. A two-hour peak period was simulated, consisting of gradually increasing private vehicle demand in three 30 minute intervals, followed by a decrease in the last 30 minutes. This was then followed by a one-hour cooldown period. Two demand scenarios were tested: a high demand scenario with an average of 32,256 vehicles entering the network and a low demand scenario with an average of 23,040 entering vehicles (Figure \ref{fig:figure_4}). To model the private vehicle routing behavior, the study utilized the stochastic c-logit route choice model integrated within AIMSUN. This routing model aimed to replicate a stochastic user-equilibrium routing solution, where vehicles select routes at the beginning of a trip to minimize travel times. 

\begin{figure}[!ht]
    \centering
        \includegraphics[trim={0.5in 0in 0.5in 0in}, clip, width=5in]{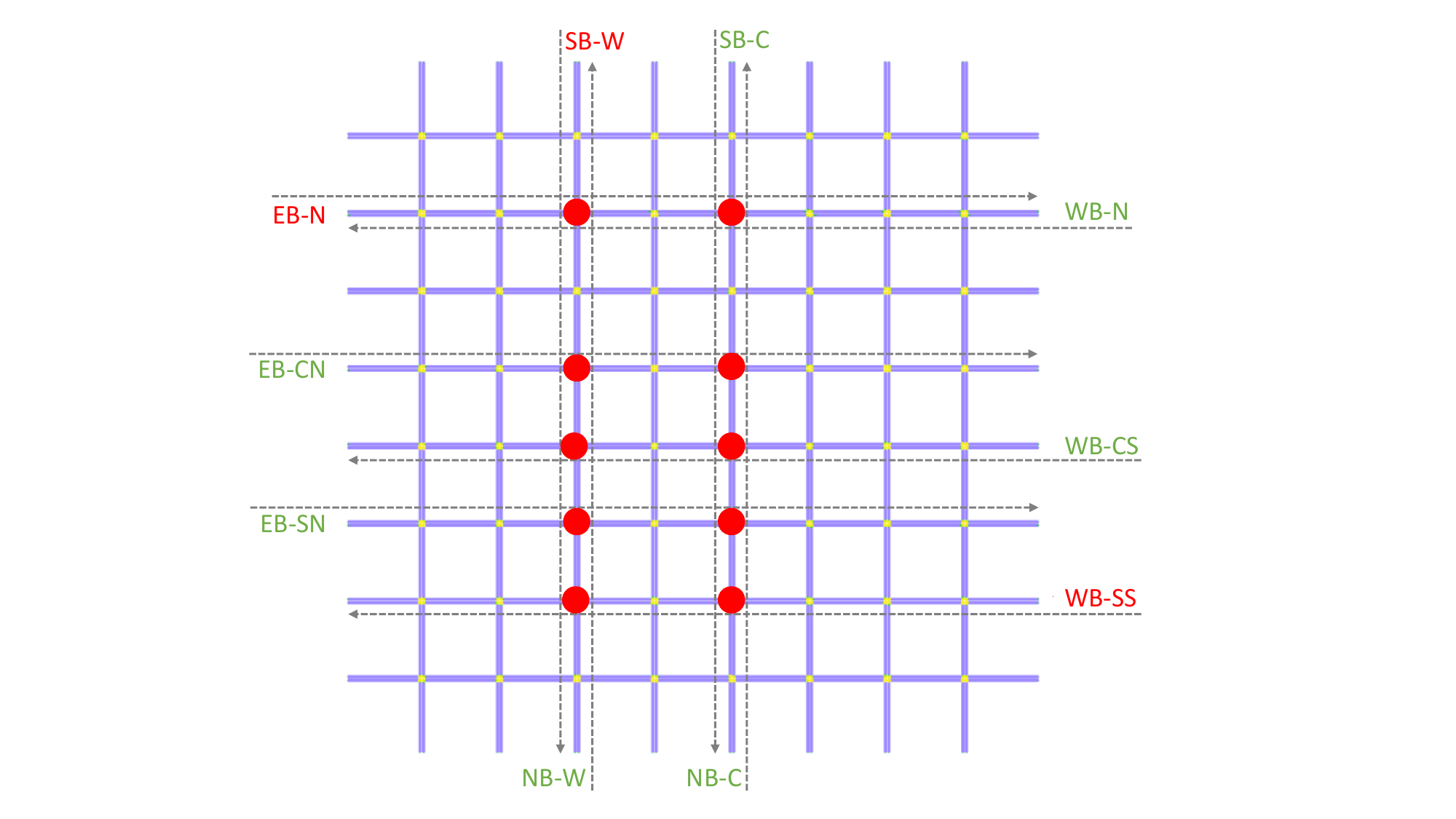}
    	\centering
    	\caption{Network configuration.}
    	\label{fig:figure_2}
\end{figure}
\begin{figure}[!ht]
    \centering
        \includegraphics[trim={0.5in 3.5in 0.5in 1.5in}, clip, width=5in]{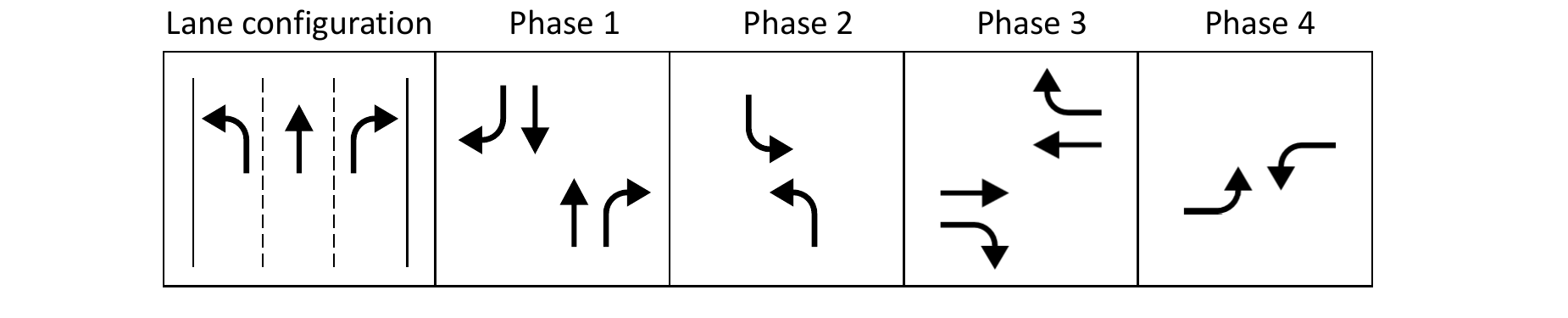}
    	\centering
    	\caption{Lane configuration and phases.}
    	\label{fig:figure_3}
\end{figure}

\begin{figure}[!ht]
    \centering
        \includegraphics[trim={0in 1in 0in 2in}, clip, width=6in]{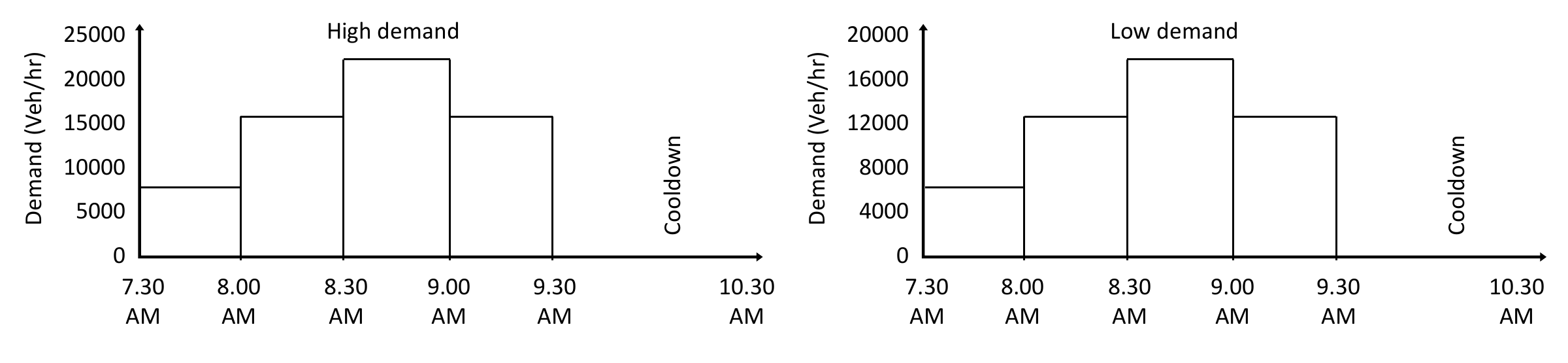}
    	\centering
    	\caption{Time varied demand.}
    	\label{fig:figure_4}
\end{figure}

The simulated network consists of ten bus routes, which include a combination of bi-directional and unidirectional routes as shown in Figure \ref{fig:figure_2}. Six of the routes operate between three pairs of O-D centroids: (SB-W, NB-W), (SB-C, NB-C), and (EB-N, WB-N). The remaining four routes -- EB-CN, WB-CS, EB-SN, and WB-SS -- are unidirectional, meaning that buses travel in only one direction. Within the network, there are seven high-occupancy routes indicated by green labels and three low-occupancy routes (marked with red labels) in Figure \ref{fig:figure_2}. The study simulates two different levels of passenger demand. In the high passenger demand scenario, the high occupancy routes are assigned an average occupancy of 50 passengers per bus, while the low occupancy routes have an average occupancy of 25 passengers per bus. In contrast, the low passenger demand scenario assumes the high occupancy routes have an average occupancy of 12 passengers per bus, while the low occupancy routes have an average occupancy of 3 passengers per bus. Two levels of bus frequencies were also simulated where the headway between buses in the high frequency case was 2 minutes on average, while the low frequency case was simulated with an average headway of 5 minutes between consecutive bus arrivals on each route. The network includes ten conflict intersections denoted by red circles in Figure \ref{fig:figure_2} where buses may compete for right of way at the same time. As conflicting movements are served by different phases, the phase with a higher pressure while considering the average upstream occupancies will be served using the OCC-MP. 

The performance of OCC-MP is compared with two other baseline policies. The first is the original Q-MP policy. The second baseline policy is an MP-based strategy that incorporates a rule-based transit signal priority, referred to as RB-MP. The RB-MP seeks to mimic the strategy proposed in \citep{xu2022a}. Specifically, it follows the MP framework to determine the weights, calculate the pressure and assign the right of way to the phase with the maximum pressure similar to the three steps shown in Section 2.1. However, the difference is that the algorithm uses a set of constraints that ensures unconditional priority to buses if they are detected within the queue of vehicles on a movement. To enforce this, the weight of a movement is increased by a sufficiently large constant, $M$, if a bus is detected which increases the pressure of the phase serving that movement. Therefore, at the end of the update interval, this phase is activated. However, if there are no buses present, RB-MP functions similar to Q-MP. In the case of multiple buses approaching an intersection on competing phases i.e., both of their movements cannot be accommodated by the same phase to avoid conflicts, the right of way is assigned to the phase that maximizes the pressure of private vehicles i.e., the phase with more queued vehicles is served. A pseudocode of the algorithm behind RB-MP is provided below.

% \begin{algorithm}[H]
% \caption{\hl{RB-MP pseudocode}}
% \SetAlgoLined
% \#calculate weight of all movements at intersection $i$

% \If{time == update interval}{

%     \For{each movement $(l,m) \in L_i^\phi$}{
%         \State $w_{\text{RB-MP}}(l,m) \leftarrow x(l,m) - \sum_{n \in D(m)} x(m,n) \times r(m,n)$\;
%         \If{$x_{\text{bus}}(l,m) > 0$}{
%             \State $w_{\text{RB-MP}}(l,m) \mathrel{+}= M$\;
%         }
%     }
% }
%     \#calculate pressure of all phases 
%     \For{each phase $\phi \in \Phi_i$}{
%         \State $P_{\text{RB-MP}}^\phi \leftarrow \sum_{(l,m) \in L_i^\phi} w_{\text{RB-MP}}(l,m) \times c(l,m)$\;
%     }

%     \#activate phase with maximum pressure
    
%         \State $S_{\text{RB-MP}} \leftarrow \arg \max_{\phi \in \Phi_i} P_{\text{RB-MP}}^\phi$\;
%         \State \Return{$S_{\text{RB-MP}}$}\;

% \end{algorithm}

% \begin{tcolorbox}[colback=yellow!80, colframe=white]
\begin{algorithm}
\caption{RB-MP pseudocode}
\begin{algorithmic}
\State \# calculate weight of all movements at intersection $i$
\If{time == update interval}
    \For{each movement $(l,m) \in L_i^\phi$}
        \State $w_{\text{RB-MP}}(l,m) \leftarrow x(l,m) - \sum_{n \in D(m)} x(m,n) \times r(m,n)$
        \If{$x_{\text{bus}}(l,m) > 0$}
            \State $w_{\text{RB-MP}}(l,m) \mathrel{+}= M$
        \EndIf
    \EndFor

\State \# calculate pressure of all phases
\For{each phase $\phi \in \Phi_i$}
    \State $P_{\text{RB-MP}}^\phi \leftarrow \sum_{(l,m) \in L_i^\phi} w_{\text{RB-MP}}(l,m) \times c(l,m)$
\EndFor

\State \# activate phase with maximum pressure
\State $S_{\text{RB-MP}} \leftarrow \arg \max_{\phi \in \Phi_i} P_{\text{RB-MP}}^\phi$
\EndIf
\State \Return{$S_{\text{RB-MP}}$}
\end{algorithmic}
\end{algorithm}
% \end{tcolorbox}

To ensure consistency in the evaluation, all three MP control policies adopt a signal update interval of 10 seconds.

\subsection{Scenario setup}
Different scenarios were simulated to understand the benefits and potential application of the proposed control policy. Scenario 1 assumes the system has no knowledge of private vehicle passenger occupancy. In this case, an average of 1.5 persons per private vehicle, as reported in \citep{schrank2021a}, is assumed. However, the exact bus occupancies are assumed to be available from APC data. This scenario is further extended to test the resilience of the policy due to errors in the reported bus occupancies from APC. To test this, a random error term was added to the occupancies of buses reported to the controller after crossing every intersection. The error term was assumed to have a mean of 0 and standard deviation of $\sigma \%$ of the true occupancy at each intersection and additive over every intersection. Varying values of $\sigma$ from 0 to 40 were tested to understand the impact of discrepancies in transmitted APC data and how it impacts the network performance. 

Scenario 2 assumes individual vehicle occupancies are available to the signal controller, as would be possible in a CV environment. This means that the system has complete knowledge of both private vehicle and bus occupancies, which is leveraged by the OCC-MP policy to dynamically calculate weights of movements based on their occupancy levels. Within the simulation, each private vehicle entering the network was randomly assigned an occupancy based on a probability distribution (shown in Table \ref{table:table_1}) such that the average private vehicle occupancy was approximately 1.5. In a fully connected environment, it is assumed that all vehicles are equipped with CV technology that is leveraged by the MP policies to accurately measure the queue lengths and (or) occupancies. However, a network may have mixed flow comprising of both connected and non-connected vehicles. Therefore, a partially connected environment was also considered in which the CV penetration rate was varied from 20\% to 100\% to understand how the policies perform when limited information is available. 

\begin{table}[htbp]
  \centering
  \caption{Probability distribution of private vehicle occupancy}\label{table:table_1}
    \begin{tabular}{cc}
    \toprule
    Occupancy & Probability \\
    \midrule
    1     & 0.7 \\
    2     & 0.125 \\
    3     & 0.1 \\
    4     & 0.05 \\
    5     & 0.025 \\
    \bottomrule
    \bottomrule
    \end{tabular}%
\end{table}%

Within both scenarios, a total of 8 sub-scenarios were simulated, each representing a different combination of private vehicle demand, bus occupancy, and bus frequency.  The private vehicle demand represents the overall traffic flow in the network, while the bus occupancy and frequency directly affect the bus operations and interactions with other vehicles. By considering both high and low occupancy levels and varying bus headways, the impact of different bus configurations on the performance of the policies can be analyzed. Table \ref{table:table_2} contains the configuration of the sub-scenarios. Each sub-scenario was simulated with 10 different random seeds to account for stochasticity and ensure robust analysis. 

\begin{table}[htbp]
  \centering
  \caption{Summary of sub-scenarios}\label{table:table_2}
    \begin{tabular}{cp{11.275em}p{10.32em}p{7.41em}}
    \toprule
    \multicolumn{1}{c}{\multirow{2}[2]{*}{Sub-scenario}} & \multirow{2}[2]{*}{Private vehicle demand} & \multirow{2}[2]{*}{Bus passenger demand} & \multirow{2}[2]{*}{Bus frequency} \\
          & \multicolumn{1}{c}{} & \multicolumn{1}{c}{} & \multicolumn{1}{c}{} \\
    \midrule
    1     & Low   & High  & High \\
    2     & Low   & High  & Low \\
    3     & Low   & Low   & High \\
    4     & Low   & Low   & Low \\
    5     & High  & High  & High \\
    6     & High  & High  & Low \\
    7     & High  & Low   & High \\
    8     & High  & Low   & Low \\
    \bottomrule
    \bottomrule
    \end{tabular}%
  \label{tab:addlabel}%
\end{table}%

\section{Results}\label{sec:results}
\subsection{Scenario 1: Non-connected vehicle environment } 
To quantify the level of congestion in the network, average network speeds under the Q-MP policy across the eight sub-scenarios are provided in Figure \ref{fig:figure_5}. The lines represent the mean value across all ten simulation iterations, while the shaded areas represent the confidence interval with $\pm$ one standard error of observed values. Sub-scenarios with similar private vehicle and bus demands but different occupancies were grouped together as the Q-MP does not consider vehicle occupancies. As expected, the average network speeds drop drastically from about 25 km/h to just under 20km/h due to the increase in private vehicle demand. A change in bus headway from 5 minutes to 2 minutes also results in a slight decrease in network speeds, as expected. 

\begin{figure}[!ht]
    \centering
        \includegraphics[trim={0.1in 0.3in 0.5in 0.5in}, clip, width=4in]{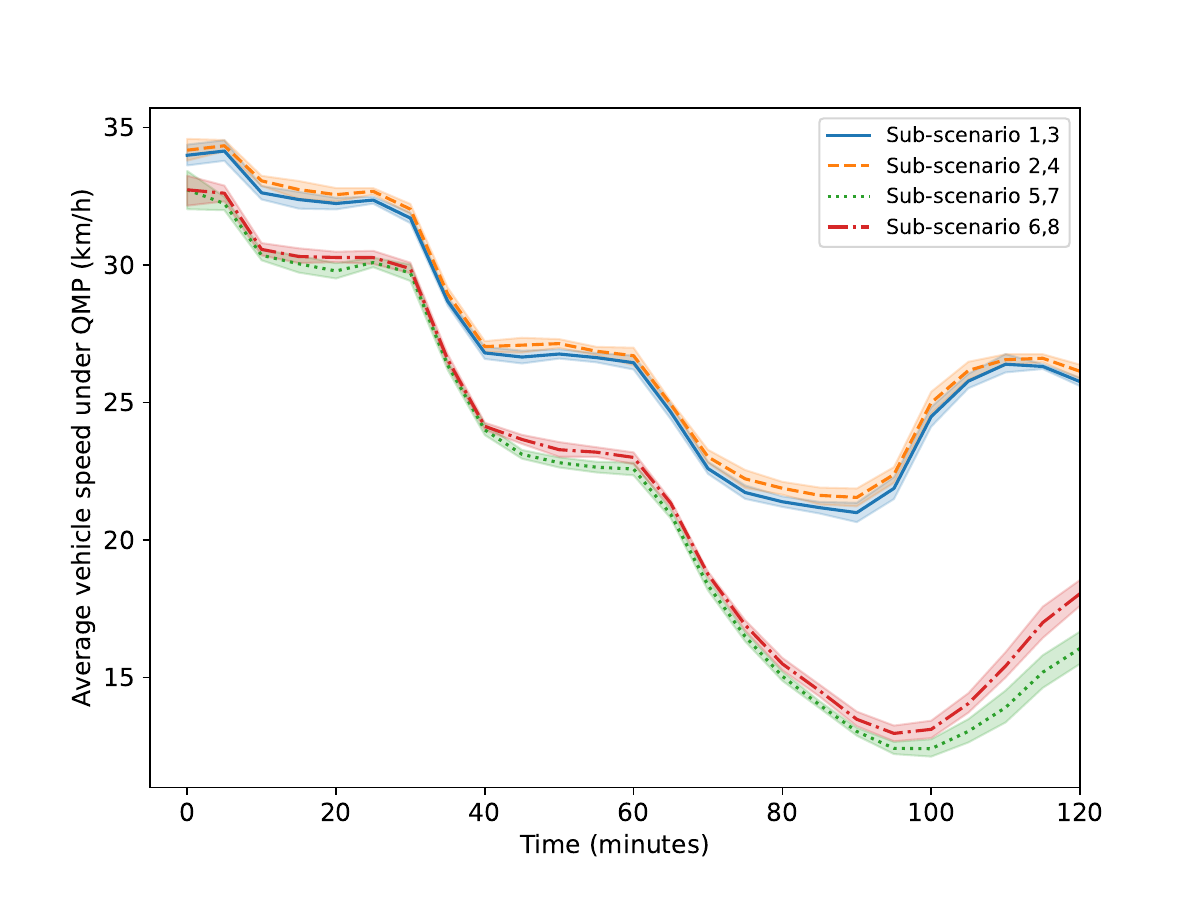}
    	\centering
    	\caption{Average vehicle speeds under Q-MP across sub-scenarios.}
    	\label{fig:figure_5}
\end{figure}

\subsubsection{Vehicle travel time comparison}
First, tests were conducted for $\sigma=0$, which indicate that APC data from buses is perfectly accurate. Figure \ref{fig:figure_6}  presents the percentage change in vehicle travel time (VTT) of private vehicles under OCC-MP and RB-MP strategies, relative to the Q-MP. Standard errors across the ten simulation iterations are shown using whiskers. It is evident that integrating TSP using either the RB-MP or OCC-MP policies results in an increase in VTT of private vehicles over Q-MP. However, OCC-MP has a lower negative impact on private vehicles compared to RB-MP across all sub-scenarios. 

\begin{figure}[!ht]
    \centering
        \includegraphics[trim={0.4in 0.2in 0.5in 0.5in}, clip, width=4in]{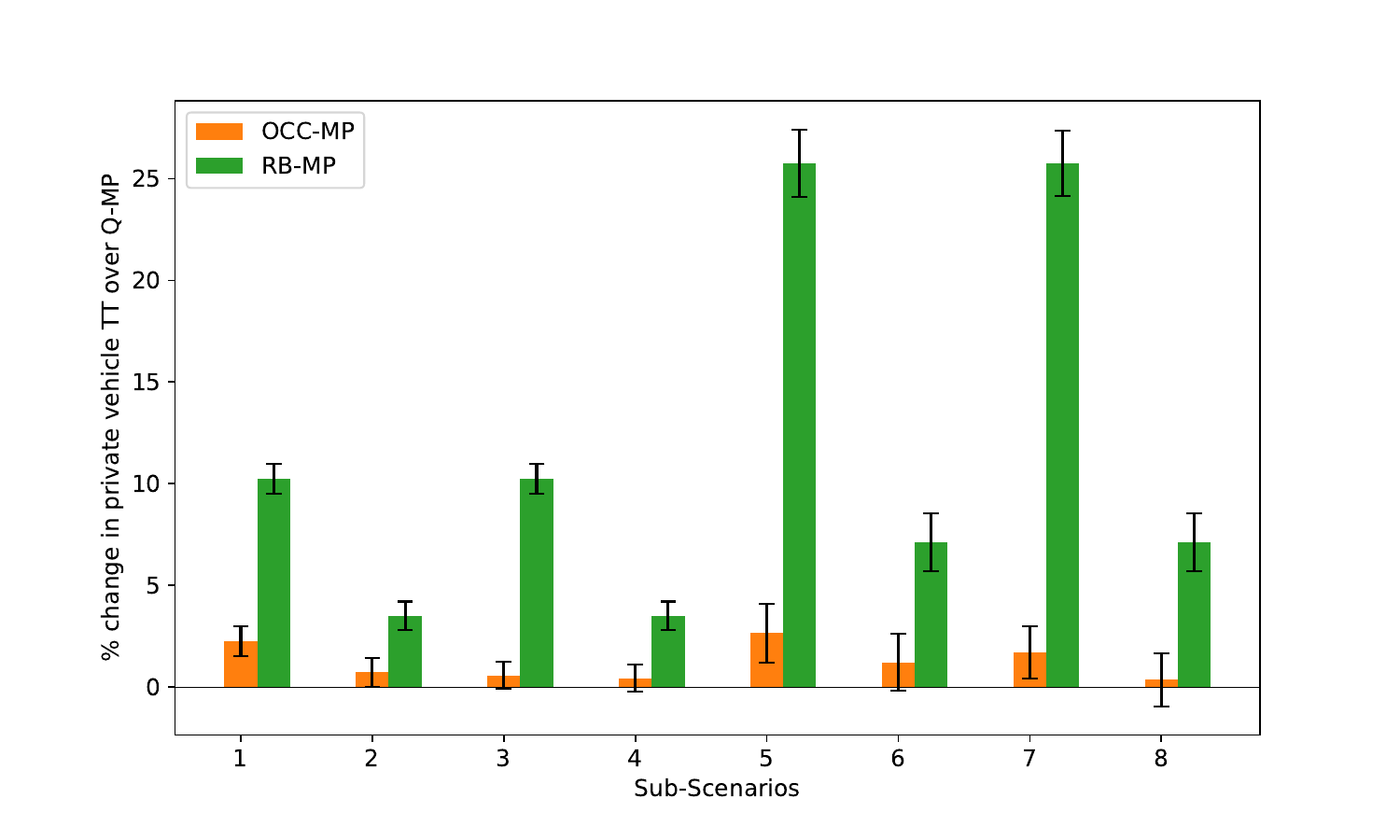}
    	\centering
    	\caption{Percentage change in private vehicle travel time over Q-MP.}
    	\label{fig:figure_6}
\end{figure}

It is expected that OCC-MP will behave similar to Q-MP when few buses are present;   Sub-scenarios 3, 4, 6 and 8 confirm this as the confidence intervals designated by the standard errors contain 0, which suggests no statistically significant difference between the performance of OCC-MP and Q-MP. Sub-scenario 8, in which the demand for private vehicles was high and buses had a lower frequency and lower passenger occupancies, resulted in only 0.36\% increase in private vehicle travel time. This can be attributed to the fact that there were fewer buses with lower occupancies in the network, leading OCC-MP to select similar phases to Q-MP. The maximum percentage change in VTT from OCC-MP is 2.64\% and observed for Sub-scenario 5, which has high bus frequency and passenger occupancy. In this sub-scenario, OCC-MP frequently selected phases to prioritize the movement of buses carrying more passengers. Note that RB-MP is not impacted by bus occupancies; thus, the same average VTTs were observed across pairs of sub-scenarios with the same vehicular demand. Overall, RB-MP resulted in statistically significant increases in VTT, ranging from 3.50\% to 25.75\%. Interestingly, the best performance of RB-MP is still worse than the worst performance of OCC-MP. This can be attributed to the fact that OCC-MP may select phases in which private vehicle queues are large, even when buses are present. The results highlight the effectiveness of the OCC-MP strategy in mitigating the negative impact on private vehicle travel times when compared to RB-MP.
To further demonstrate the difference in impacts to private vehicles across the three control strategies, Figure \ref{fig:figure_7} plots the VTT of private vehicles per minute. Notice that the VTT continues to grow under the RB-MP policy despite the reduction in vehicle demand at the 90 minute mark. By contrast, the Q-MP and OCC-MP policies show VTT trends that reflect the level of vehicle demand. This finding is indicative of queue spillback phenomenon due to growing vehicle queues in the RB-MP policy. 

\begin{figure}[!ht]
    \centering
        \includegraphics[trim={0.2in 0.1in 0.5in 0.4in}, clip, width=3.5in]{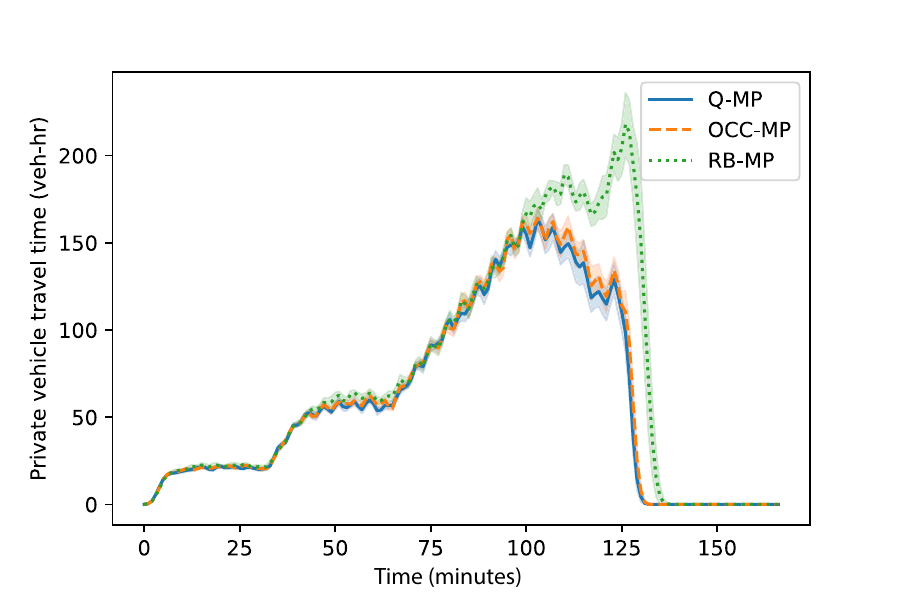}
    	\centering
    	\caption{Total travel time of private vehicles (Sub-scenario 7).}
    	\label{fig:figure_7}
\end{figure}

Figure \ref{fig:figure_8} illustrates the percent change in bus VTT under both OCC-MP and RB-MP compared to Q-MP. The results show that both strategies lead to a reduction in bus travel times compared to the baseline Q-MP strategy across all sub-scenarios, and all improvements are statistically significant. However, the magnitude of the improvement varies between the two strategies. As expected, RB-MP consistently outperforms OCC-MP and provides larger reductions in bus VTT since it provides full priority to buses. OCC-MP achieves an average reduction in bus VTT of 14.5\% when buses have higher occupancies (Sub-Scenarios 1, 2, 5 and 6) and 7.5\% when buses are less crowded (Sub-Scenarios 3, 4, 7 and 8). This is expected as weights of bus movements are lower when there are fewer passengers onboard. Conversely, RB-MP shows little variation between the different sub-scenarios and achieves a nearly consistent average reduction of approximately 30\% across all sub-scenarios. 

\begin{figure}[!ht]
    \centering
        \includegraphics[trim={0.4in 0.2in 0.5in 0.5in}, clip, width=4in]{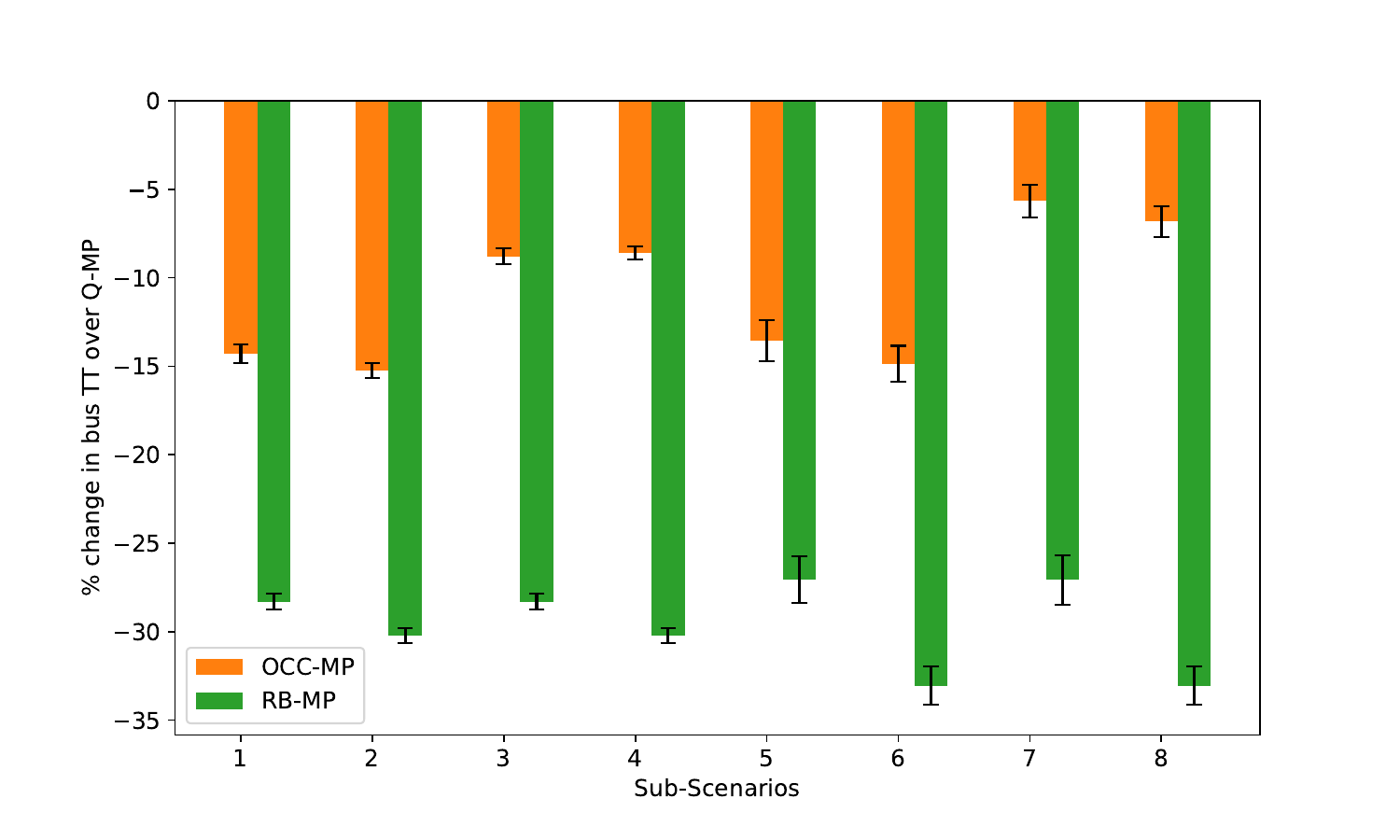}
    	\centering
    	\caption{Percentage change in bus travel time over Q-MP.}
    	\label{fig:figure_8}
\end{figure}

\subsubsection{Passenger travel time comparison}
Figure \ref{fig:figure_9} compares the passenger travel times (PTT) of both OCC-MP and RB-MP strategies against Q-MP for all sub-scenarios. The results reveal that OCC-MP yields lower total PTT compared to Q-MP in 6 out of 8 sub-scenarios, indicating a positive impact on overall passenger mobility. The improvements range from approximately 0.1\% to 3.6\% on average, mostly in scenarios with higher bus occupancies. Maximum benefits were observed in Sub-scenario 1, suggesting OCC-MP best reduces overall passenger travel times when there are relatively fewer private vehicles and more buses carrying more passengers. Sub-scenarios 7 and 8 saw a nominal increase in passenger travel times by 0.9\% and 0.1\% respectively over Q-MP; however, the confidence intervals denoted by the standard errors reveal these increases are note statistically significant. Conversely, RB-MP shows mixed results with some sub-scenarios exhibiting improvements and others significant negative effects on PTT over Q-MP. Sub-scenarios 1 and 2 show improvements of 3.5\% and 2.3\% respectively, which were similar to OCC-MP in terms of PTT improvements. However, in the other sub-scenarios, RB-MP results in an increase in PTT ranging from approximately 1.9\% up to 21.2\% in Sub-Scenario 7. Previously it was found that Sub-Scenario 7 also corresponds to the highest increase in VTT of private vehicles and lowest bus VTT improvement compared to Q-MP. This highlights the superior performance of OCC-MP in balancing VTT of private vehicles and buses, ultimately resulting in lower passenger travel times.

\begin{figure}[!ht]
    \centering
        \includegraphics[trim={0.4in 0.2in 0.5in 0.5in}, clip, width=4in]{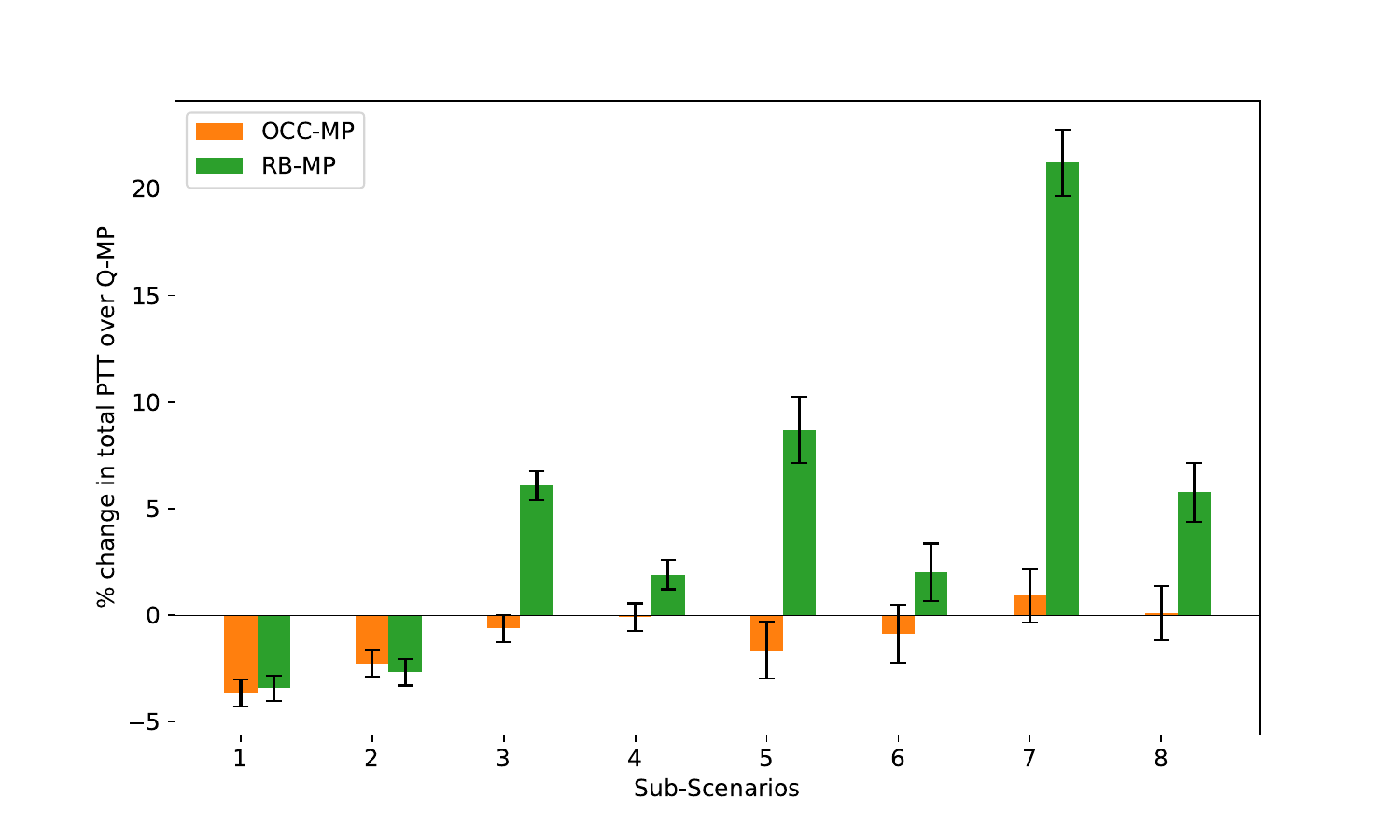}
    	\centering
    	\caption{Percent change in total passenger travel time over Q-MP.}
    	\label{fig:figure_9}
\end{figure}

\subsubsection{Network stability comparison}
Sub-section \ref{subsec:max_stability} of this paper analytically proves that OCC-MP possesses maximum stability at an isolated intersection. To compare the stability property of the three control algorithms used in this study on an urban network, simulation results of the average accumulation (i.e., number of vehicles in the network) over time across the sub-scenarios is shown in Figure \ref{fig:figure_10}. Vehicular demand is relatively low for the first 30 minutes of the simulation, where it can be observed that the average accumulation is similar for all three control policies. However, as the demand increases, the network experiences a higher average accumulation for RB-MP than Q-MP or OCC-MP. Between 60-90 minutes, the vehicular demand is highest and represents an unstable condition for all three control policies specifically in sub-scenarios 5-8 as the average network accumulation keeps growing over time. 
 
 As Q-MP has the largest stable region, Figure\ref{fig:figure_11} compares the difference in accumulation for OCC-MP and RB-MP to further investigate how the stable region differs from Q-MP for different sub-scenarios. It is evident that OCC-MP performs similar to Q-MP for most of the sub-scenarios as the shaded regions include 0. The accumulation under OCC-MP is higher than Q-MP in sub-scenarios 1 and 5 between 60-120 minutes of the simulation. This is expected as the sudden spike in private vehicle demand leads to an increase in accumulation while buses have priority. The other sub-scenarios are not associated with a significant difference over Q-MP. However, for the same private vehicle demand patterns, RB-MP has a significantly higher average accumulation compared to OCC-MP when buses arrive at higher frequencies as seen in Sub-scenarios 1, 3 and 5, 7. This suggests that the RB-MP has a smaller stable region compared to Q-MP and OCC-MP while OCC-MP exhibits a similar stable region as Q-MP in an urban network for private vehicles, even while prioritizing buses. 

\begin{figure}[!ht]
    \centering
    \begin{subfigure}{0.45\textwidth}
        \includegraphics[width=2.7in]{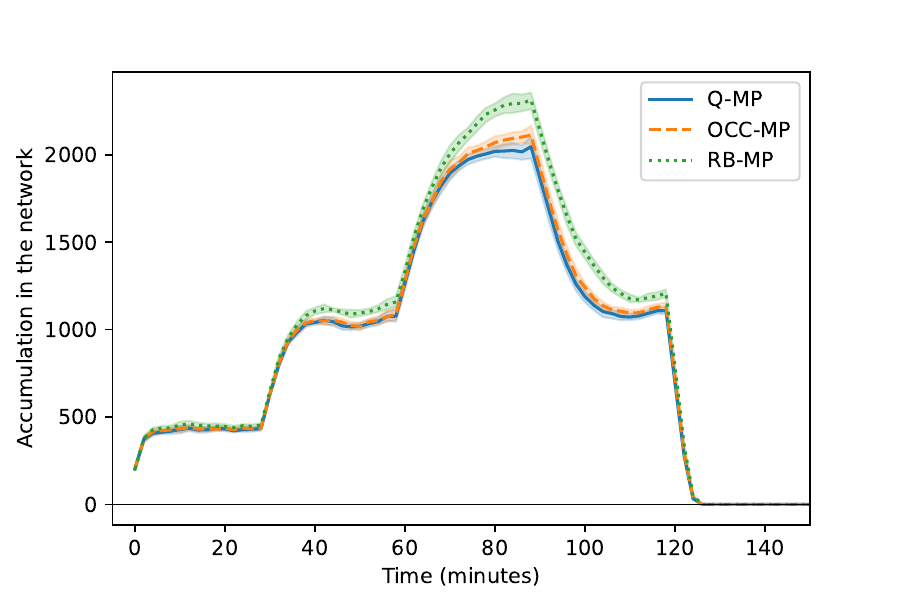}
    	\centering
    	\caption{Sub-scenario 1}
    	\label{fig:figure_10_0}
    \end{subfigure}
    \begin{subfigure}{0.45\textwidth}
        \includegraphics[width=2.7in]{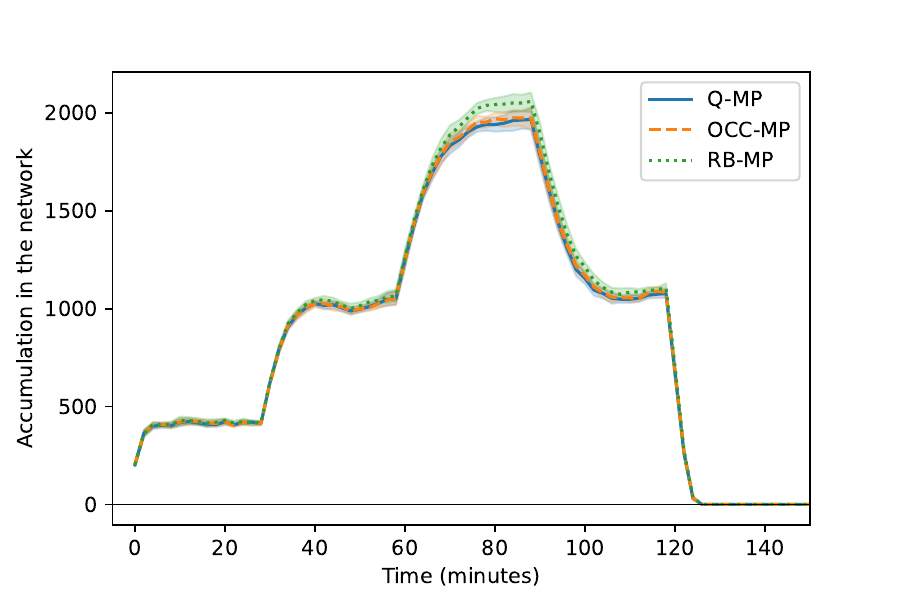}
    	\centering
    	\caption{Sub-scenario 2}
    	\label{fig:figure_10_1}
    \end{subfigure}

    \begin{subfigure}{0.45\textwidth}
        \includegraphics[width=2.7in]{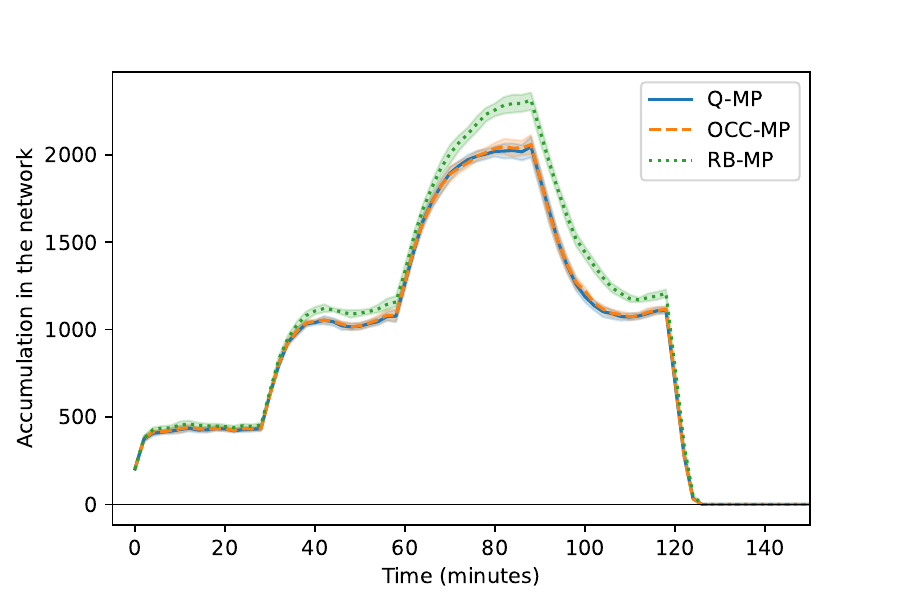}
    	\centering
    	\caption{Sub-scenario 3}
    	\label{fig:figure_10_2}
    \end{subfigure}
    \begin{subfigure}{0.45\textwidth}
        \includegraphics[width=2.7in]{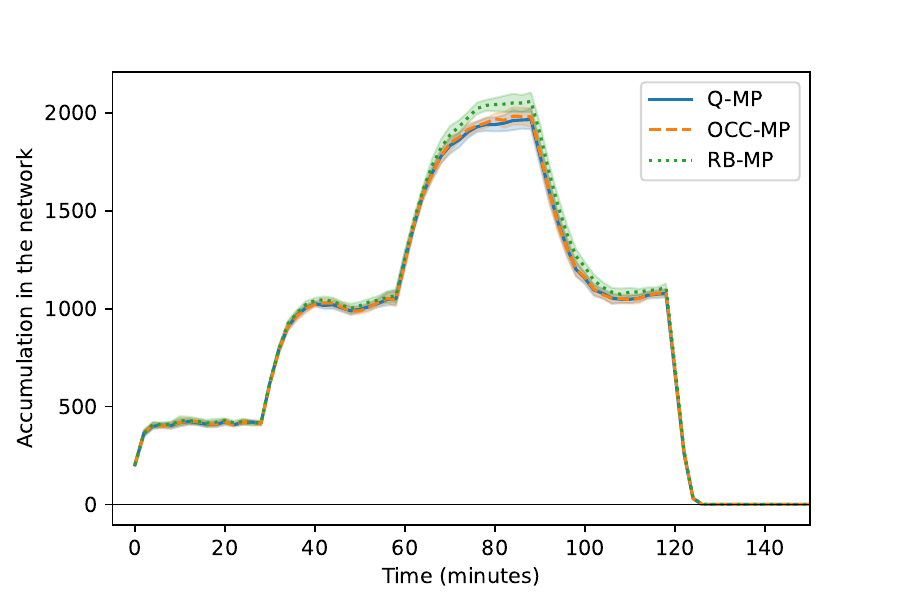}
    	\centering
    	\caption{Sub-scenario 4}
    	\label{fig:figure_10_3}
    \end{subfigure}

    \begin{subfigure}{0.45\textwidth}
        \includegraphics[width=2.7in]{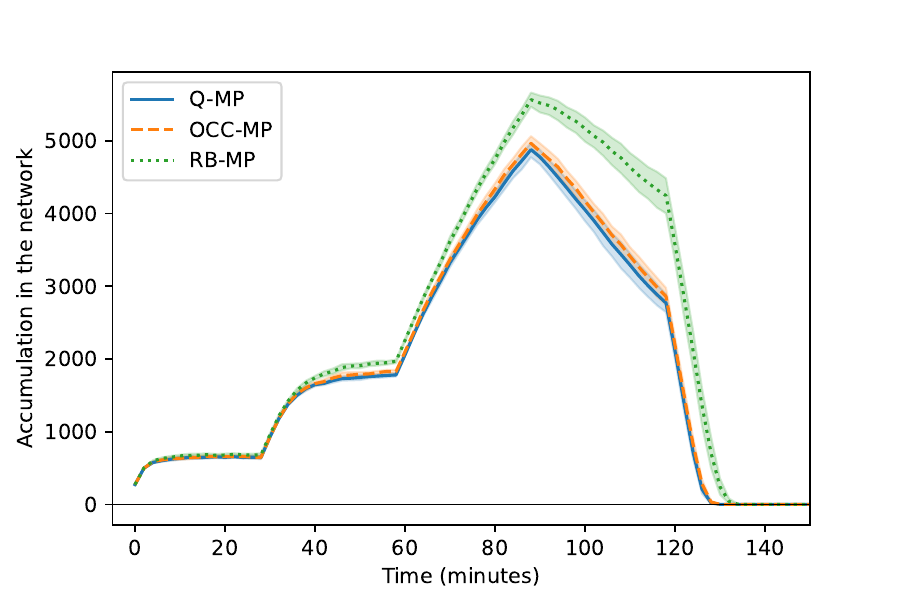}
    	\centering
    	\caption{Sub-scenario 5}
    	\label{fig:figure_10_4}
    \end{subfigure}
    \begin{subfigure}{0.45\textwidth}
        \includegraphics[width=2.7in]{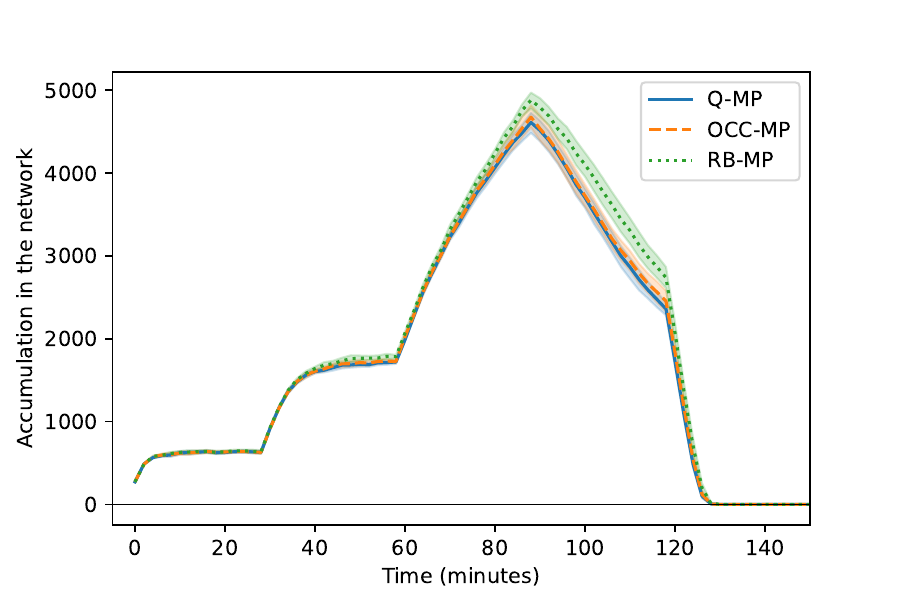}
    	\centering
    	\caption{Sub-scenario 6}
    	\label{fig:figure_10_5}
    \end{subfigure}

    \begin{subfigure}{0.45\textwidth}
        \includegraphics[width=2.7in]{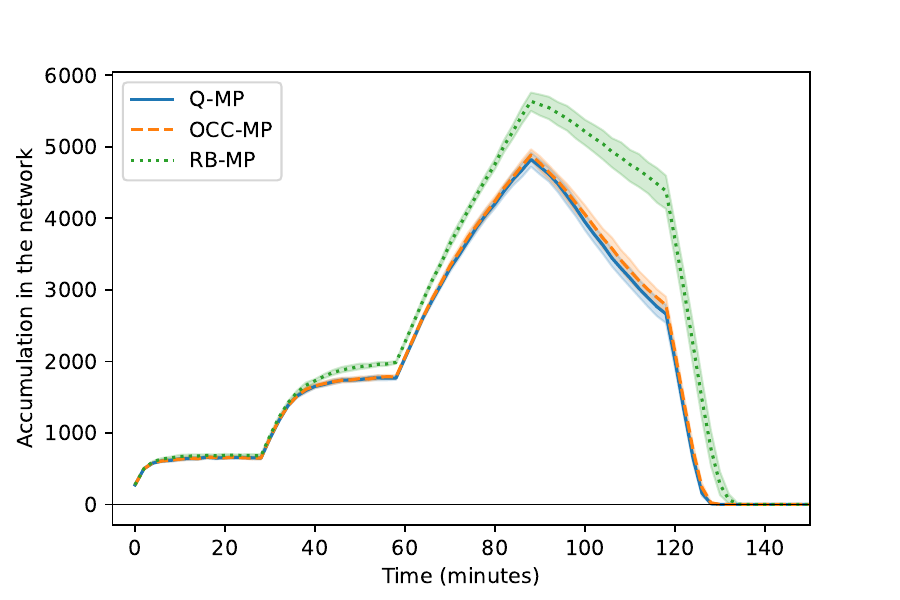}
    	\centering
    	\caption{Sub-scenario 7}
    	\label{fig:figure_10_6}
    \end{subfigure}
    \begin{subfigure}{0.45\textwidth}
        \includegraphics[width=2.7in]{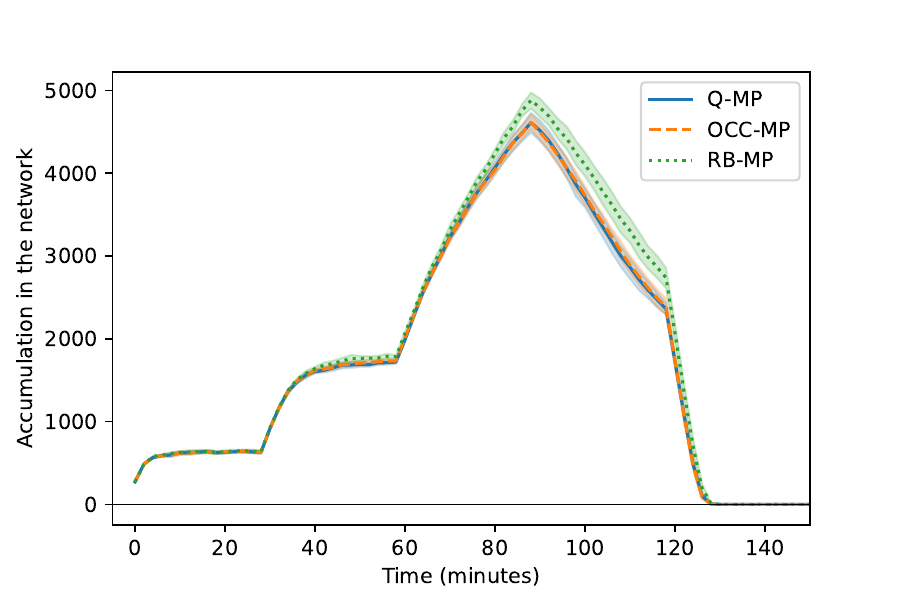}
    	\centering
    	\caption{Sub-scenario 8}
    	\label{fig:figure_10_7}
    \end{subfigure}
    
    \caption{Evolution of average accumulation in the network under different control policies for Sub-scenarios 1-8.}
    \label{fig:figure_10}
\end{figure}

\begin{figure}[!ht]
    \centering
    \begin{subfigure}{0.45\textwidth}
        \includegraphics[width=2.7in]{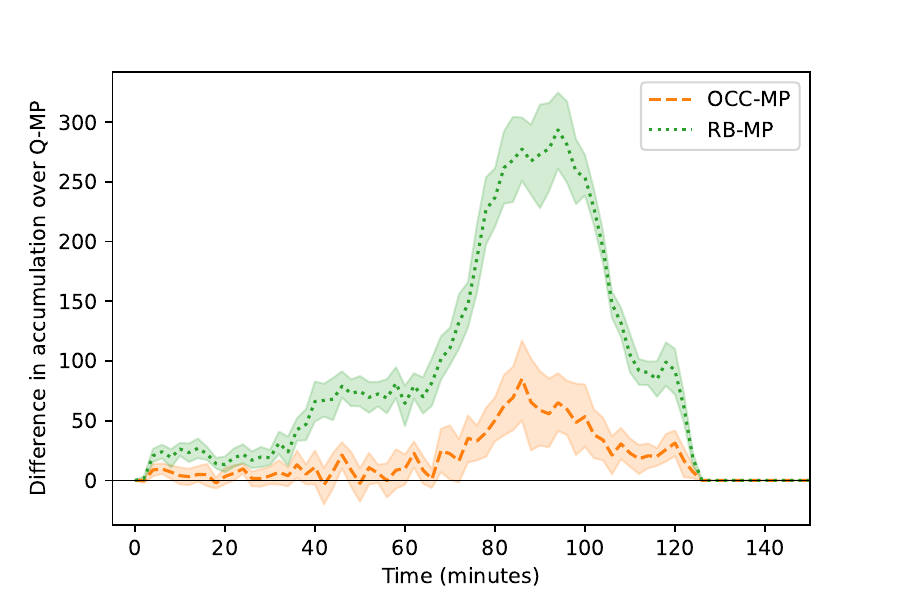}
    	\centering
    	\caption{Sub-scenario 1}
    	\label{fig:figure_11_0}
    \end{subfigure}
    \begin{subfigure}{0.45\textwidth}
        \includegraphics[width=2.7in]{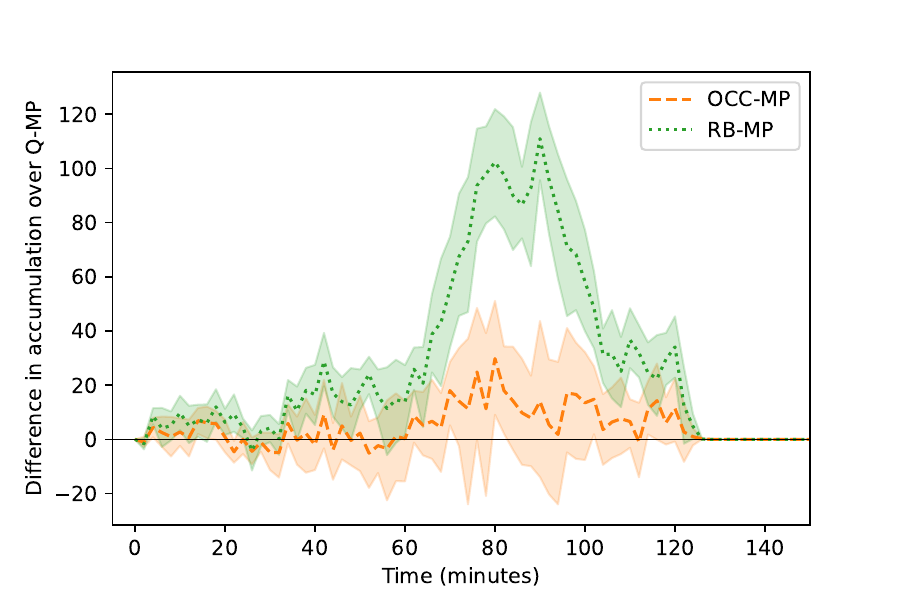}
    	\centering
    	\caption{Sub-scenario 2}
    	\label{fig:figure_11_1}
    \end{subfigure}

    \begin{subfigure}{0.45\textwidth}
        \includegraphics[width=2.7in]{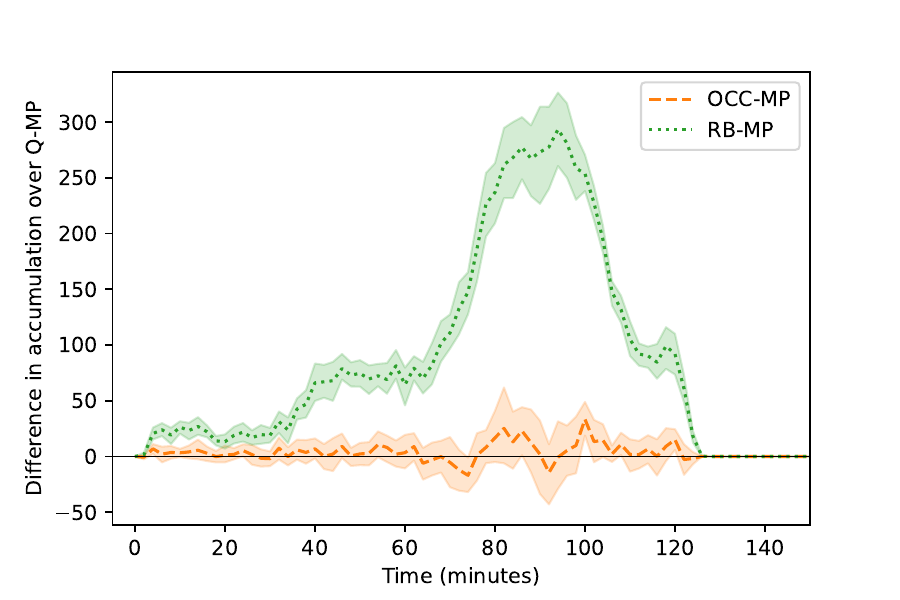}
    	\centering
    	\caption{Sub-scenario 3}
    	\label{fig:figure_11_2}
    \end{subfigure}
    \begin{subfigure}{0.45\textwidth}
        \includegraphics[width=2.7in]{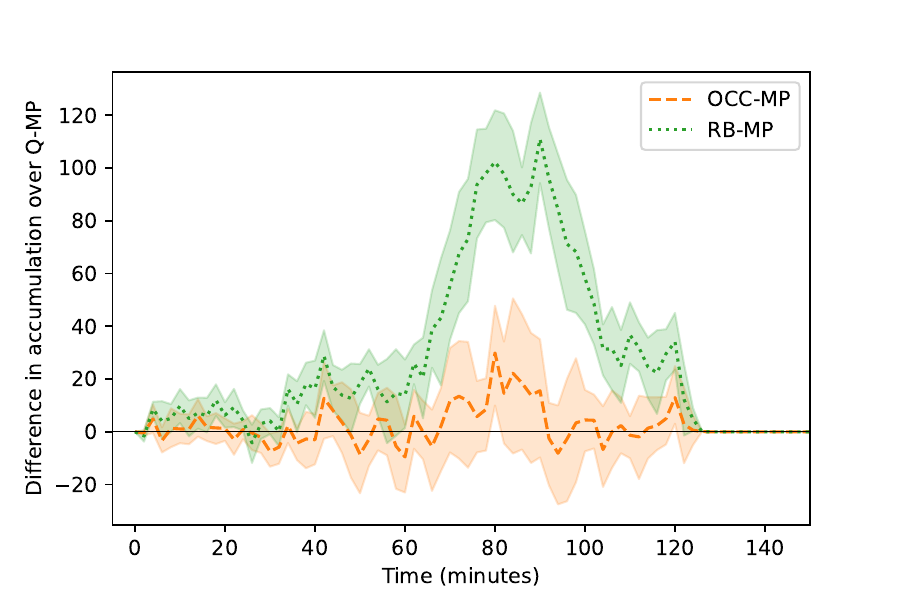}
    	\centering
    	\caption{Sub-scenario 4}
    	\label{fig:figure_11_3}
    \end{subfigure}

    \begin{subfigure}{0.45\textwidth}
        \includegraphics[width=2.7in]{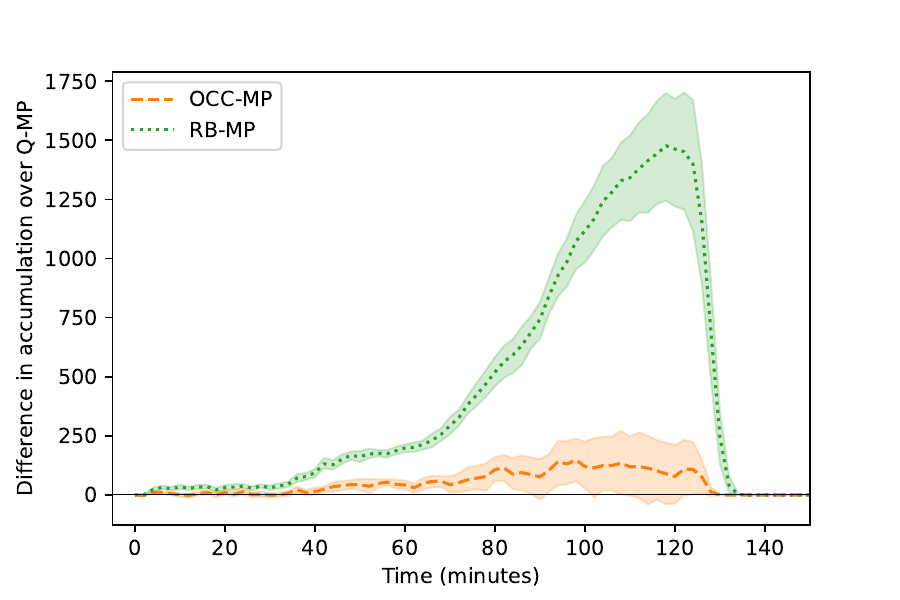}
    	\centering
    	\caption{Sub-scenario 5}
    	\label{fig:figure_11_4}
    \end{subfigure}
    \begin{subfigure}{0.45\textwidth}
        \includegraphics[width=2.7in]{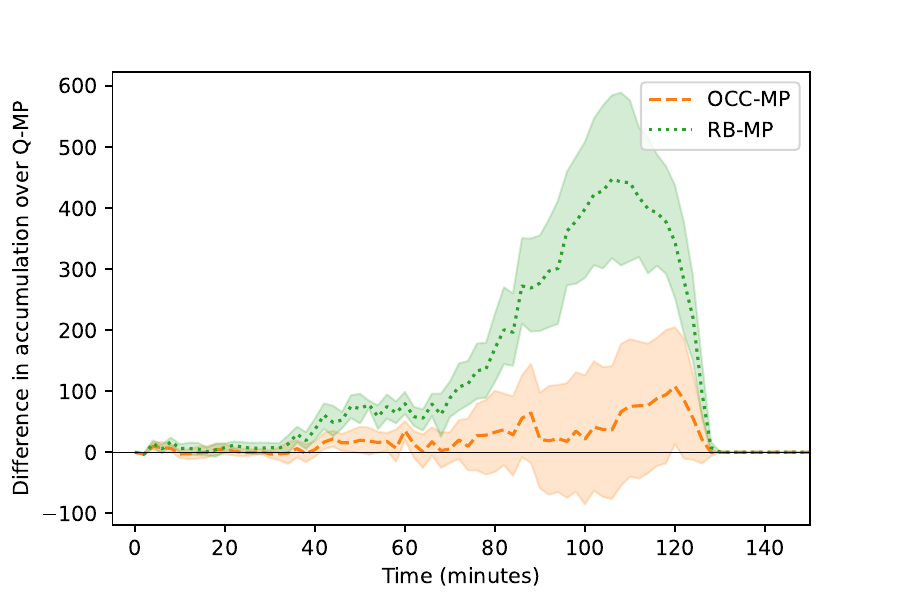}
    	\centering
    	\caption{Sub-scenario 6}
    	\label{fig:figure_11_5}
    \end{subfigure}

    \begin{subfigure}{0.45\textwidth}
        \includegraphics[width=2.7in]{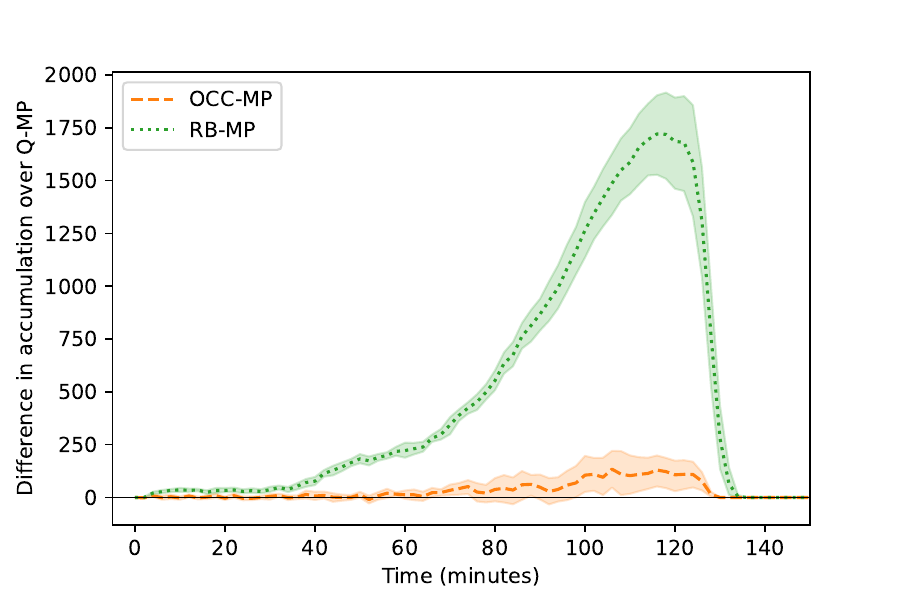}
    	\centering
    	\caption{Sub-scenario 7}
    	\label{fig:figure_11_6}
    \end{subfigure}
    \begin{subfigure}{0.45\textwidth}
        \includegraphics[width=2.7in]{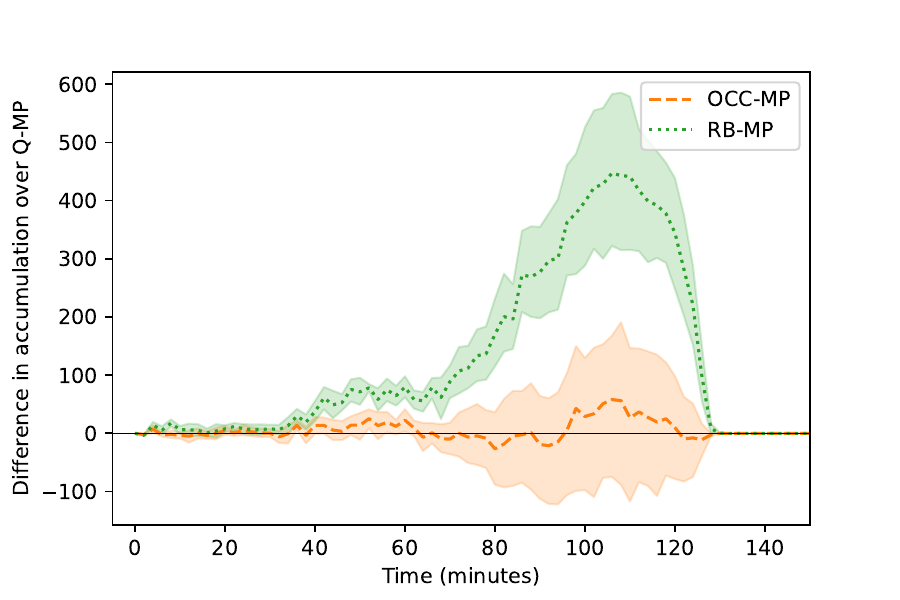}
    	\centering
    	\caption{Sub-scenario 8}
    	\label{fig:figure_11_7}
    \end{subfigure}
    
    \caption{Evolution of average accumulation in the network under different control policies for sub-scenarios 1-8.}
    \label{fig:figure_11}
\end{figure}
\subsubsection{Sensitivity to APC errors}
The proposed OCC-MP policy relies on accurate bus passenger occupancy for signal updates. To assess robustness of the policy against errors in APC data, additional simulations were conducted for sub-scenarios 1 and 3 in which  incremental error terms of $\pm$ $\sigma$ \% were introduced to the real occupancy of buses. Specifically, the occupancy being transmitted to the controller at each intersection was subject to this error and continued to accumulate until the bus exits the network.

Table \ref{table:table_3} shows the VTT of private vehicles and buses and combined PTT as $\sigma$ was increased from 0 to 40 at each intersection. The results indicate that there is relatively little variation observed across the network performance metrics. From  Table \ref{table:table_3}a, it can be seen that the travel times of private vehicles, buses, and passengers in Sub-scenario 1 do not change significantly with $\sigma$. This suggests that OCC-MP performs reasonably well even when there is significant misreporting of bus occupancies, particularly for buses with high passenger demand. Sub-scenario 3 corresponds to a case with similar private vehicle and bus demand as Sub-scenario 1 but with fewer bus passengers. Therefore, it is expected that further underreporting of its occupancy may give it little to no priority over private vehicles and potentially result in higher travel times.  Table \ref{table:table_3}b shows that change in $\sigma$ leads to slightly higher travel times than when $\sigma=0$. However, all values except bus travel time at $\sigma=40$ remain within one standard error of $\sigma=0$, suggesting differences are not statistically significant. The consistency of the results indicates that the OCC-MP policy can effectively adapt to and optimize travel times under potential discrepancies in the APC data.
\begin{table}[!htbp] \centering 
  \caption{Summary of network performance against variance in APC data } \label{table:table_3}
  \label{} 
\begin{tabular}{p{2cm}p{2cm}p{1.8cm}p{2cm}p{1.8cm}p{2cm}p{1.8cm}}
\\[-1.8ex]\hline 
\hline \\ \multicolumn{7}{c}{\textbf{(a) Sub-scenario 1}} \\ 
\cline{1-7} \\
[-1.8ex] & \multicolumn{2}{c}{Private vehicle}  & \multicolumn{2}{c}{Bus}  & \multicolumn{2}{c}{All passengers} \\
\cline{2-7} \\
$\sigma$ & Travel time (veh-hr) & Standard error & Travel time (veh-hr) & Standard error & Travel time (pax-hr) & Standard error
\\ 
\hline \\[-1.8ex] 
 \textbf{0\%} & 2298.66 & 12.33 & 37.39 & 0.14  & 5035.75 & 24.06 \\
    \textbf{10\%} & 2294.46 & 10.74 & 37.21 & 0.13  & 5021.77 & 19.81 \\
    \textbf{20\%} & 2295.34 & 11.42 & 37.14 & 0.11  & 5020.11 & 20.42 \\
    \textbf{30\%} & 2299.86 & 11.13 & 37.09 & 0.08  & 5023.77 & 18.53 \\
    \textbf{40\%} & 2299.46 & 12.22 & 37.32 & 0.14  & 5034.65 & 23.16 \\ 
\hline 
\hline \\
\multicolumn{7}{c}{\textbf{(b) Sub-scenario 3}} \\ 
\cline{1-7} \\
[-1.8ex] & \multicolumn{2}{c}{Private vehicle}  & \multicolumn{2}{c}{Bus}  & \multicolumn{2}{c}{All passengers} \\
\cline{2-7} \\
$\sigma$ & Travel time (veh-hr) & Standard error & Travel time (veh-hr) & Standard error & Travel time (pax-hr) & Standard error
\\ 
\hline \\[-1.8ex] 
 \textbf{0\%} & 2260.89 & 10.06 & 39.79 & 0.1   & 3754.45 & 15.72 \\
    \textbf{10\%} & 2265.1 & 9.93  & 39.88 & 0.1   & 3761.63 & 15.56 \\
    \textbf{20\%} & 2263.54 & 10.54 & 39.81 & 0.09  & 3758.36 & 16.36 \\
    \textbf{30\%} & 2262.01 & 9.98  & 39.81 & 0.12  & 3756.06 & 16.04 \\
    \textbf{40\%} & 2269.44 & 10.8  & 39.93 & 0.1   & 3768.99 & 16.57 \\ 
\hline 
\hline \\
\end{tabular} 
\end{table}

\subsection{Scenario 2: Connected vehicle environment}
 \subsubsection{Fully connected environment}

The OCC-MP strategy was evaluated by simulating private vehicles with known occupancies and variable bus occupancies to understand how the control policy impacts travel time of non-transit HOVs. Since RB-MP does not differentiate vehicles by occupancy, it was not included in the analysis. Figure \ref{fig:figure_12} presents a comparison of the percent change in PTT for OCC-MP over Q-MP for vehicles with different vehicle occupancies; values of 1 to 5 indicate private vehicles, while 6+ refers to buses. The results reveal that single occupancy vehicles experience an increase in their travel times over the Q-MP. However, OCC-MP effectively prioritizes movements with higher occupancy vehicles, resulting in reduced travel times for those vehicles. Specifically, vehicles with an occupancy of 3 or more experience improvement in travel time in 5 out of 6 sub-scenarios. Interestingly, sub-scenarios with low private vehicle and bus demand (2 and 4), exhibit lower travel time for vehicles with occupancy of 2 and more, highlighting the positive impact of OCC-MP. By prioritizing HOV and buses even in mixed flow conditions without dedicated bus or HOV lanes, OCC-MP can serve as a strategic approach to discourage single-occupancy vehicles on the roads, promoting more efficient and sustainable transportation options.

\begin{figure}[!ht]
    \centering
        \includegraphics[trim={0.4in 0.2in 0.5in 0.5in}, clip, width=4in]{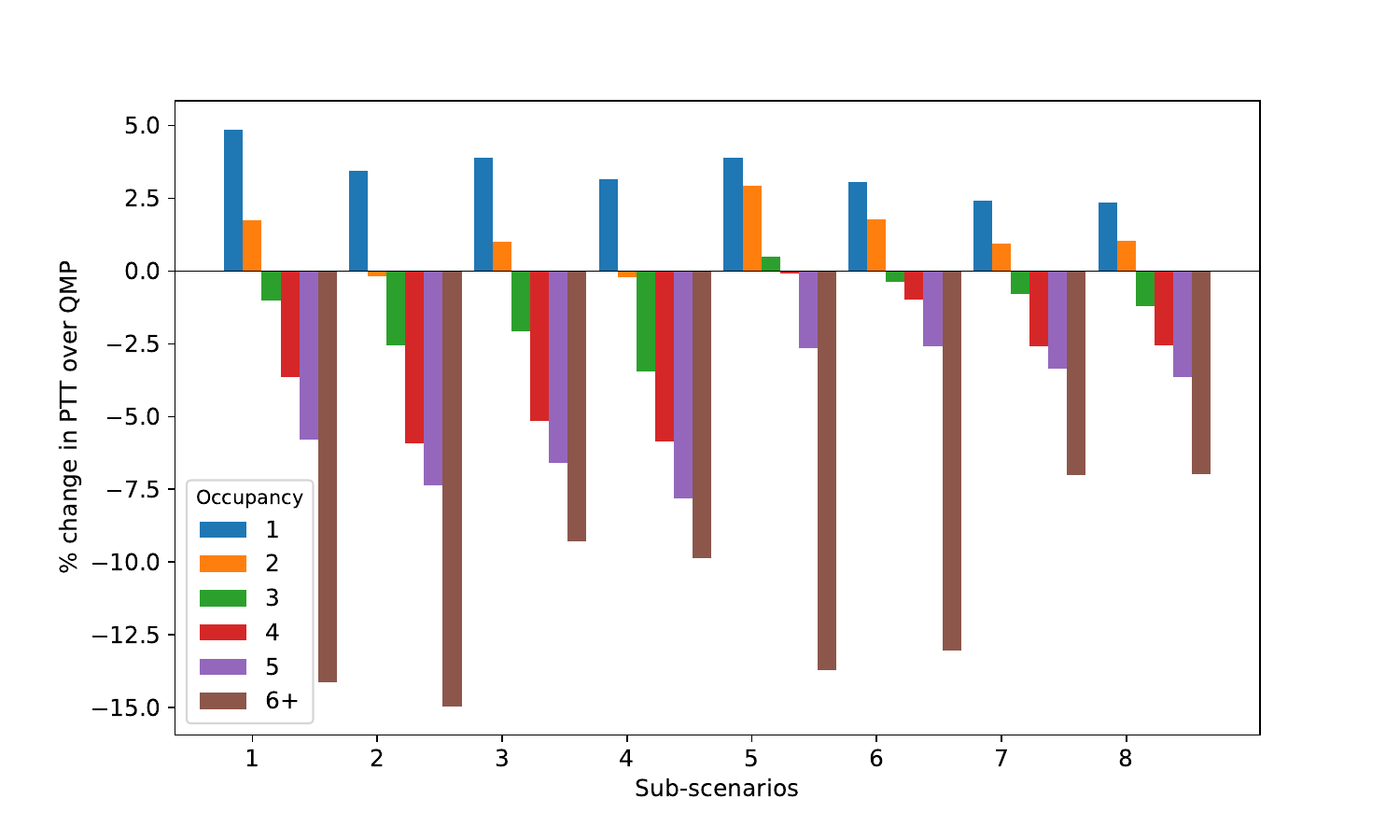}
    	\centering
    	\caption{Percent change in passenger travel time over Q-MP by occupancy under OCC-MP.}
    	\label{fig:figure_12}
\end{figure}

\subsubsection{Partially connected environment}

Although CV technology provides the potential to acquire detailed information from individual vehicles directly, implementation of a fully CV environment is farfetched. Therefore, the performances of the proposed OCC-MP policy and baseline methods were investigated under varying rates of CV penetration. In these tests, all control policies rely only on the information obtained from these CVs for measurement and updating the signal times. 
Figure \ref{fig:figure_13} shows the evolution of vehicle accumulation in the network for various CV penetration rates. The accumulation is highest for all three control policies when information from only 20\% of the private vehicles is available. With increasing CV penetration rate, the number of queued vehicles in the network drops for all three policies resulting in lower congestion. Notice, however, the returns are diminishing with respect to CV penetration rate; i.e., the highest improvements are gained from increasing the penetration rate when the penetration rate is low. Note also that both Q-MP and OCC-MP have similar performance in terms of network congestion and show consistent reduction in vehicle accumulation with increasing CV penetration rate, while the RB-MP strategy consistently performs the worst.

\begin{figure}[!ht]
    \centering
    \begin{subfigure}{0.45\textwidth}
        \includegraphics[width=2.7in]{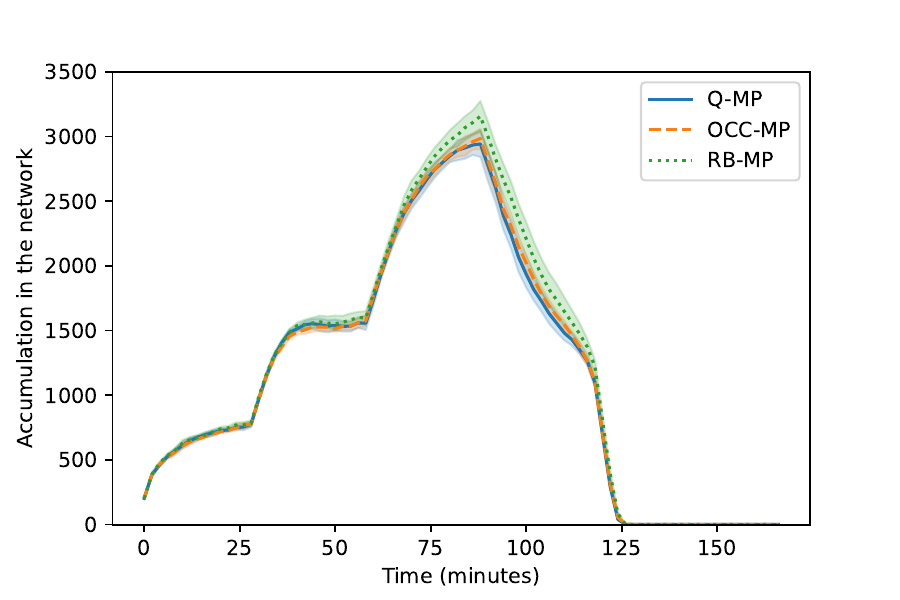}
    	\centering
    	\caption{20\%.}
    	\label{fig:figure_13_a}
    \end{subfigure}
    \begin{subfigure}{0.45\textwidth}
        \includegraphics[width=2.7in]{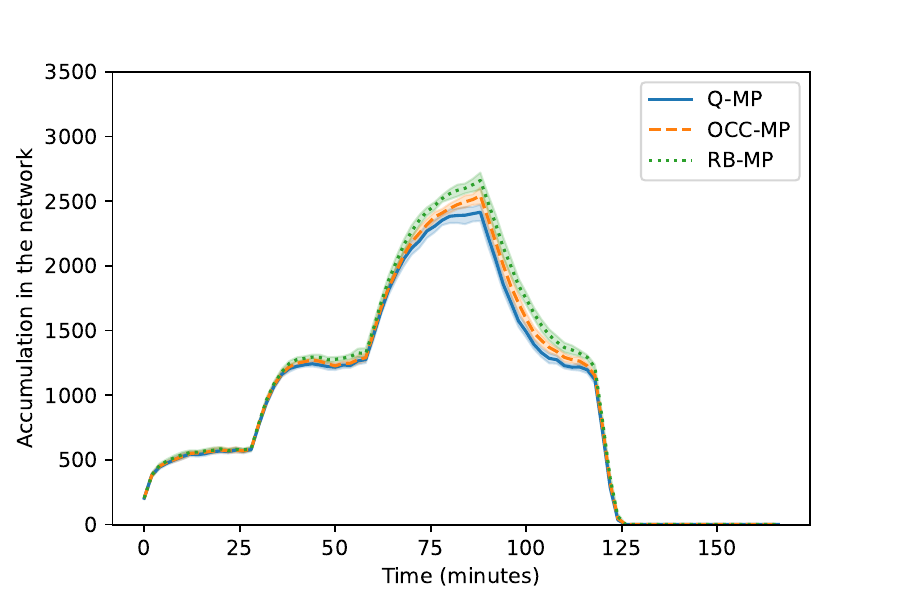}
    	\centering
    	\caption{40\%.}
    	\label{fig:figure_13_b}
    \end{subfigure}
    
    \begin{subfigure}{0.45\textwidth}
        \includegraphics[width=2.7in]{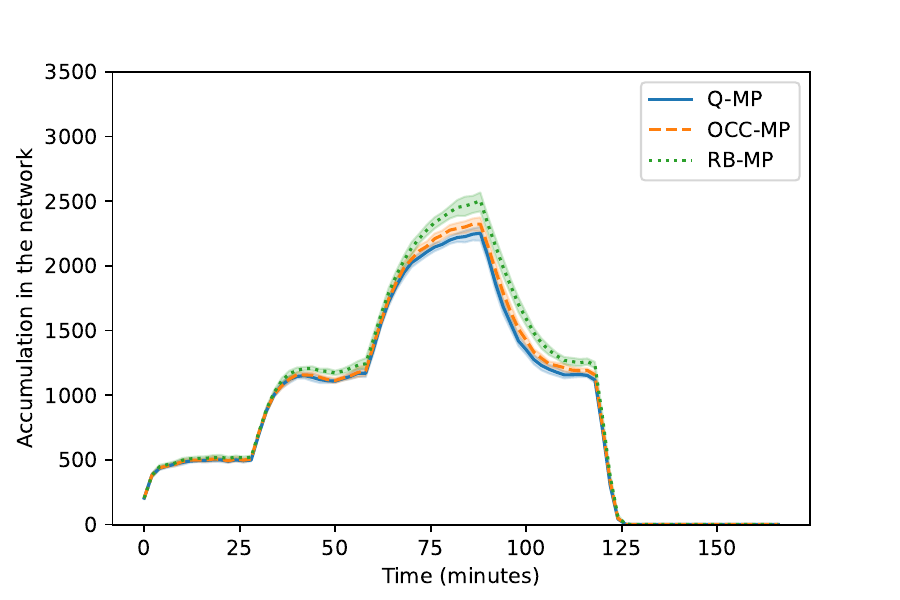}
    	\centering
    	\caption{60\%.}
    	\label{fig:figure_13_c}
    \end{subfigure}    
    \begin{subfigure}{0.45\textwidth}
        \includegraphics[width=2.7in]{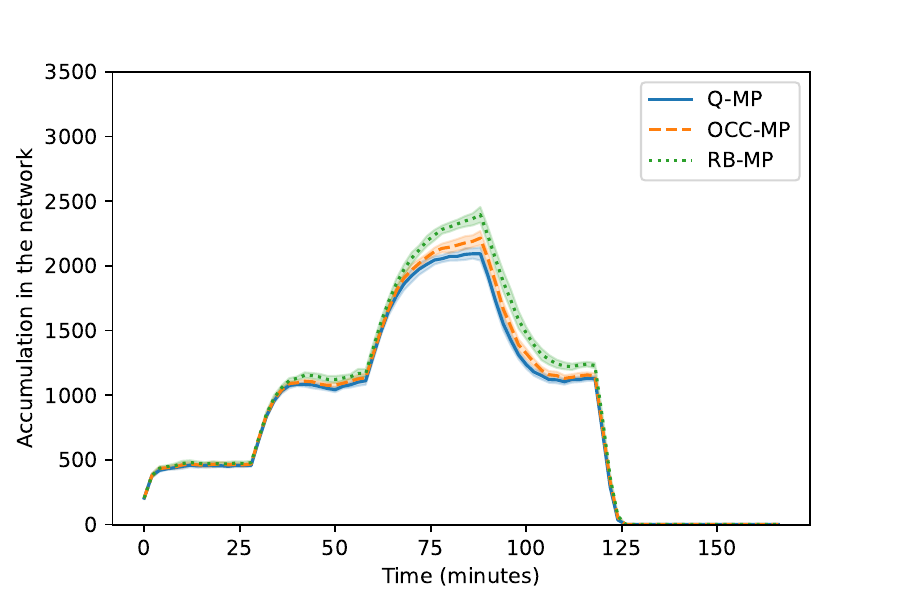}
    	\centering
    	\caption{80\%.}
    	\label{fig:figure_13_d}
    \end{subfigure}
    
    \begin{subfigure}{0.45\textwidth}
        \includegraphics[width=2.7in]{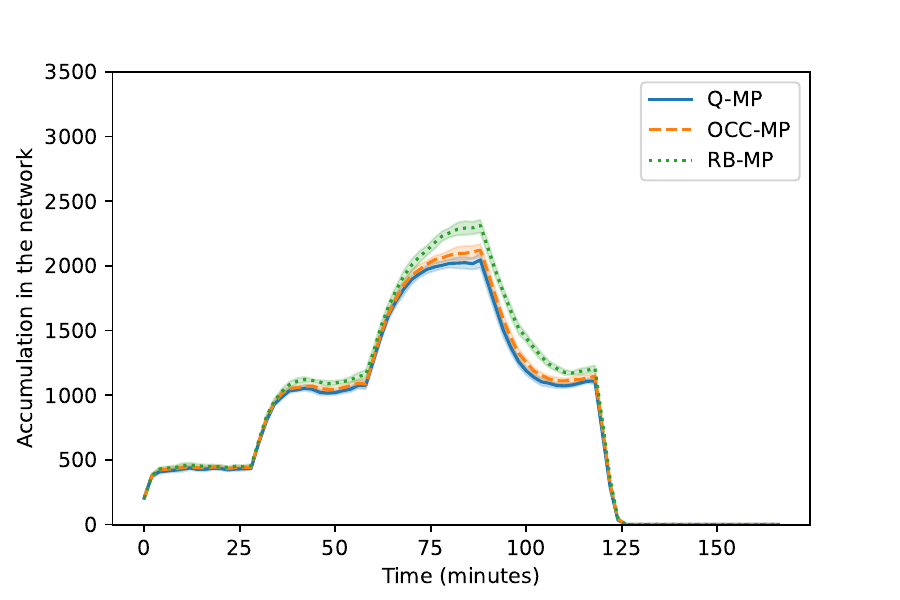}
    	\centering
    	\caption{100\%.}
    	\label{fig:figure_13_e}
    \end{subfigure}
    \caption{Accumulation of vehicles in the network for different CV penetration rates .}
    \label{fig:figure_13}
\end{figure}

The performance of the control policies in terms of private VTT, bus VTT and total PTT is shown in Figure \ref{fig:figure_14} for Sub-Scenario 1, which was chosen because OCC-MP demonstrated the largest improvement in PTT. Overall, it is observed that the increase in penetration rate of CVs improves the travel time of private vehicles and reduces the standard error across all policies, as more information becomes available on the actual queue lengths vehicle occupancy (Figure \ref{fig:figure_14_a}). The most significant improvements are observed for an increase in the penetration rate from 20\% to 40\%. This slowly diminishes as the penetration rate is further increased. RB-MP consistently demonstrates inferior performance compared to both Q-MP and OCC-MP. At 20\% penetration rate, travel times of private vehicles under Q-MP and OCC-MP are very similar, but Q-MP further reduces travel times with the increase in CV penetration.

\begin{figure}[!ht]
    \centering
    \begin{subfigure}{0.45\textwidth}
        \includegraphics[width=2.7in]{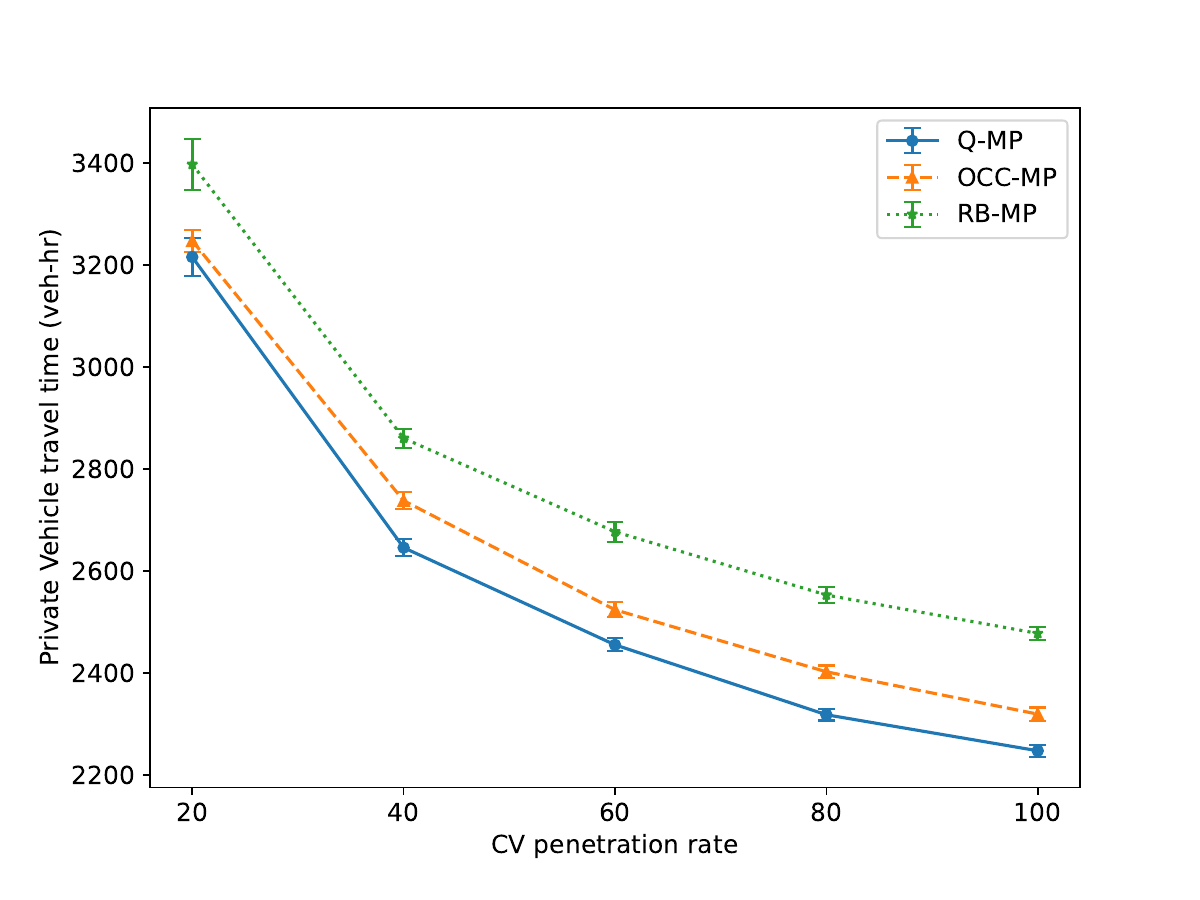}
    	\centering
    	\caption{Private vehicle travel time.}
    	\label{fig:figure_14_a}
    \end{subfigure}
    \begin{subfigure}{0.45\textwidth}
        \includegraphics[width=2.7in]{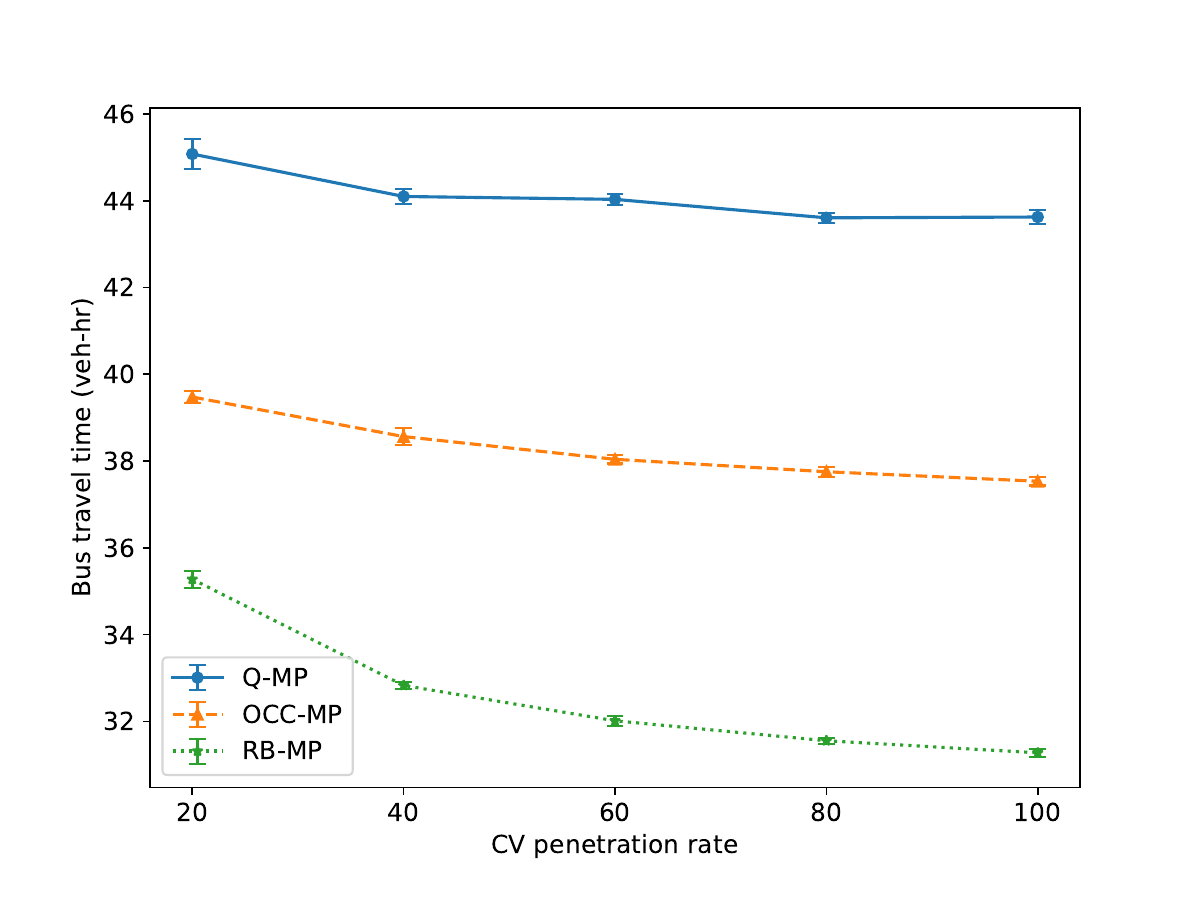}
    	\centering
    	\caption{Bus travel time.}
    	\label{fig:figure_14_b}
    \end{subfigure}
    
    \begin{subfigure}{0.45\textwidth}
        \includegraphics[width=2.7in]{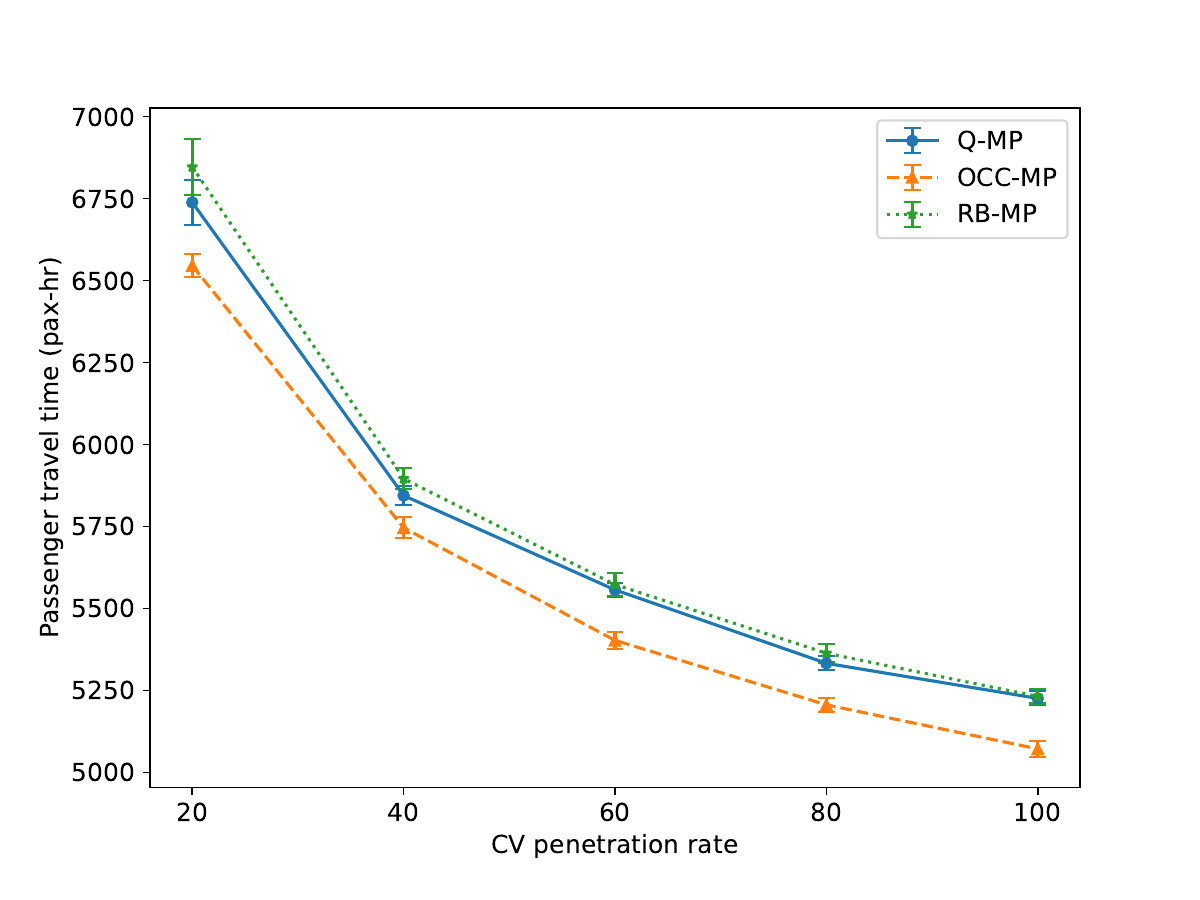}
    	\centering
    	\caption{Passenger travel time.}
    	\label{fig:figure_14_c}
    \end{subfigure}
    \caption{Effect of CV penetration rate.}
    \label{fig:figure_14}
\end{figure}

From Figure \ref{fig:figure_14_b}, it is evident that bus VTT also improves as the percentage of connected vehicles in the network increases. Although there is a tradeoff between private vehicle and bus travel times, increasing CV penetration translates to lower congestion in the network (Figure \ref{fig:figure_11}), which in turn improves the overall bus operations. Note that these improvements are nominal for both Q-MP and OCC-MP, whereas larger improvements are observed for RB-MP with lower standard errors across the random seeds. Despite resulting in higher private vehicle and bus VTT compared to Q-MP and RB-MP policies, respectively, OCC-MP consistently resulted in the lowest PTT for all CV penetration rates (Figure \ref{fig:figure_14_c}). Both the Q-MP and RB-MP have very similar performances in terms of PTT, with reductions that are smaller than that achieved by the OCC-MP. This highlights the reliability of the proposed OCC-MP policy even when only a subset of the vehicles is connected. 

\section{Conclusion}

Conventional MP algorithms rely only on vehicle-related metrics to update signal timings. Thus, these algorithms maximize vehicle throughput while overlooking the impact on transit and other HOVs. This study introduces an occupancy-based Max Pressure (OCC-MP) algorithm that considers both the number of queued vehicles and the passenger occupancies within the upstream queued vehicles. The proposed policy theoretically guarantees maximum stability at an isolated intersection, and is shown to prioritize transit and HOV movement without reducing the size of the stable region compared to the original MP when applied on an entire network. 

The performance of OCC-MP was tested against the original max pressure (Q-MP) and a rule-based MP algorithm that provides TSP (RB-MP). Micro-simulation tests on a grid network demonstrate that  OCC-MP outperforms RB-MP in terms of reducing negative impacts on private vehicles while reducing bus VTT compared to Q-MP. Overall, OCC-MP results in lower PTT under various demand and occupancy levels. This is because OCC-MP not only prioritizes transit vehicles and those with higher occupancies, but also serves the movements with large private vehicle queues. The best performance was observed in scenarios with lower private vehicle demands and higher bus demands, with larger bus passenger occupancies. A stability analysis showed that OCC-MP has a larger stable region than the RB-MP policy and one that is comparable to the Q-MP policy; this suggests that OCC-MP is able to handle similar demands to Q-MP (and larger than RB-MP) while still providing priority to buses. The control policy also demonstrates nominal variation in passenger travel time from errors in APC data, highlighting the robustness of the algorithm. Further tests in a CV environment show that an increase in the penetration rate of CVs improve the overall performance of OCC-MP in reducing PTT. In a fully CV environment, OCC-MP consistently outperforms baseline methods in reducing the VTT of HOVs and buses making it a sustainable strategy to discourage single occupancy vehicles in a transportation network without the need to implement expensive dedicated lane facilities. Moreover, OCC-MP can be readily deployed as a conditional TSP strategy in real world environments yet to fully transition to CV technology with minimal additional sensing requirements, primarily relying on existing technologies such as TSP communications and APC systems. Corridors where buses send a TSP request are already equipped with vehicle detection technology which can be slightly modified to include occupancy data transmission from buses. Its universal applicability renders it useful for implementation on networks with or without dedicated bus rapid transit (BRT) or HOV facilities. 

 Although the simulations were conducted on a grid network, further research can explore the performance of OCC-MP in more complex urban networks. Since the applicability of the proposed OCC-MP encompasses mixed traffic, it may be interesting to explore its performance in networks with dedicated bus lanes or HOV lanes. Moreover, given the increasing emphasis on creating "complete streets" that accommodate various modes of transportation, future studies may consider developing MP control algorithms that consider the complexities of multimodal transport. It is worth noting that \citep{liu2023total} demonstrated that different MP algorithms may have different optimal update intervals that maximize their performance. Therefore, the impact of optimal time-step for signal update interval can be explored for OCC-MP.

\section{Acknowledgements}
This research was supported by NSF, United States Grant CMMI-1749200.

\section{Author Contributions}
The authors confirm contribution to the paper as follows: study conception and design: TA, HL, VG; analysis and interpretation of results: TA, HL, VG; draft manuscript preparation: TA, HL, VG. All authors reviewed the results and approved the final version of the manuscript.

\newpage

\bibliographystyle{elsarticle-harv}
\bibliography{biblio}
\end{document}